\documentclass[a4paper,12pt]{article}

\usepackage[margin=1in]{geometry}
\usepackage{amsmath,amssymb,amsthm,mathrsfs}
\usepackage{graphicx}
\usepackage{booktabs}
\usepackage{array}
\usepackage{makecell}
\usepackage{threeparttable}
\usepackage{subfig}
\usepackage{microtype}
\usepackage{placeins}
\usepackage[hidelinks]{hyperref}
\hypersetup{%
  pdftitle={A Probability-Based Four-Moment Framework for Basket and Spread Option Pricing under Correlated Lognormal Models: An Empirical Crack-Spread Application},
  pdfauthor={Dongdong Hu, Hasanjan Sayit, Steve Tchoneteck, Frederi Viens},
  pdfkeywords={basket options, spread options, crack spread, correlated lognormal variables, four-moment matching}
}

\graphicspath{{figures/}}
\allowdisplaybreaks[2]

\theoremstyle{plain}
\newtheorem{theorem}{Theorem}[section]
\newtheorem{proposition}[theorem]{Proposition}
\newtheorem{lemma}[theorem]{Lemma}
\newtheorem{corollary}[theorem]{Corollary}

\theoremstyle{definition}
\newtheorem{definition}[theorem]{Definition}

\theoremstyle{remark}
\newtheorem{remark}[theorem]{Remark}

\newcommand{\ud}{\mathrm{d}}
\newcommand{\R}{\mathbb{R}}
\newcommand{\PP}{\mathbb{P}}
\newcommand{\E}{\mathbb{E}}
\newcommand{\Var}{\operatorname{Var}}
\newcommand{\Corr}{\operatorname{Corr}}
\newcommand{\1}{\mathbf{1}}
\newcommand{\rowrule}{\specialrule{0.3pt}{2pt}{2pt}}

\title{Beyond Lognormal Sums: A Four-Moment Probability Framework for Basket and Spread Option Pricing }

\author{%
\begin{tabular}{c}
Dongdong Hu$^{1}$ and Hasanjan Sayit$^{2}$\\
Steve Tchoneteck$^{3}$ and Frederi Viens$^{3}$\\[0.45em]
\small $^{1}$Yiwu Industrial \& Commercial College, \texttt{hudongdong@ywicc.edu.cn}\\
\small $^{2}$Xi'an Jiaotong-Liverpool University, \texttt{Hasanjan.Sayit@liverpool.ac.uk}\\
\small $^{3}$Department of Statistics, Rice University, \texttt{st83@rice.edu}, \texttt{fv15@rice.edu}
\end{tabular}%
}
\date{}

\begin{document}
\maketitle

\begin{abstract}
Basket options are difficult to value under correlated lognormal dynamics
because weighted sums and differences of lognormal variables have no tractable
distribution. This paper develops a probability-based four-moment framework
that separates the exact pricing representation from the distributional
approximation. A change of measure first writes a basket price as a linear
combination of probabilities. For a standard basket with one positive weight,
these probabilities become CDF values of positive correlated lognormal sums.
Each sum is approximated by a shifted lognormal variance mixture matched to its
first four moments. For an unrestricted mixed-sign basket, a signed shifted
lognormal proxy gives an analytical call-price formula. We state admissibility
conditions, provide a practical root-selection rule, establish the main
strike-based financial properties of the direct proxy, and derive exact
pricing-error identities in terms of cumulative distribution function (CDF) discrepancies. The numerical analysis combines standard-basket benchmarks with an empirical application to a normalized $3{:}2{:}1$ crack spread constructed from RBOB
gasoline, ULSD or heating oil, and WTI futures. The results show that the probability
reformulation and the fourth-moment condition improve the distributional fit
and pricing accuracy, particularly when maturity and tail asymmetry increase.
The framework remains analytical, transparent, and suitable for repeated
valuation across strikes and maturities.
\end{abstract}

\noindent\textbf{Keywords:} basket options; standard-basket options; spread
options; crack spread; correlated lognormal variables; four-moment matching;
shifted lognormal approximation; analytical option pricing.

\medskip
\noindent\textbf{JEL classifications:} C02, C58, G13
\newpage
\section{Introduction}
\label{sec:introduction}

Basket options are written on weighted portfolios rather than on single
assets. They are useful when an economic exposure depends jointly on several
prices, as in equity-index baskets, currency portfolios, agricultural
processing margins, and energy spreads. In many such contracts, some assets
enter with positive weights and others with negative weights. This structure
is economically natural, but it complicates valuation. Under the
Black-Scholes model, a European option on one asset has a closed-form price
\cite{blackscholes1973,merton1973}. The same argument does not extend to a
basket because a sum or difference of correlated lognormal variables is not
lognormal. Its distribution is generally unavailable in closed form, and the
exercise region is nonlinear. Monte Carlo simulation provides a reliable
benchmark \cite{glasserman2004monte}, but repeated simulation can be expensive
when prices and sensitivities are needed for many strikes, maturities, or
parameter values. Analytical approximations remain useful when they are fast,
transparent, and sufficiently accurate for repeated valuation.

A broad literature addresses this problem. Margrabe derived an exact formula
for a zero-strike exchange option \cite{margrabe1978value}, while Kirk's
approximation is widely used for two-asset spreads with a nonzero strike
\cite{kirk1995correlation}. Carmona and Durrleman developed approximations based
on the geometry of the exercise boundary
\cite{carmona2003pricing,carmona2005generalizing}. Related closed-form and
boundary-expansion methods appear in \cite{li2008closed,Deng2008closed}, Fourier
methods in \cite{dempster2002spread}, and lower-dimensional approximations in
\cite{Caldana_Fusai20134893,bjerksund2014closed}. Distributional methods that
allow negative basket weights are studied in \cite{borovkova2007}. A second
line of work uses conditioning. Curran conditions on a geometric mean
\cite{curran1994valuing}, with related constructions in
\cite{rogers1995value,vorst1992prices,deelstra2003pricing}. Taylor expansions
also lead to useful formulas for arithmetic-average and basket options
\cite{ju2002}. Although these methods differ in implementation, they address
the same basic difficulty: the payoff depends on a dependent sum of lognormal
variables.

Moment matching provides a direct distributional alternative. The classical
Fenton--Wilkinson approximation matches the first two moments
\cite{fenton1960sum}, and Schwartz and Yeh develop a recursive approximation
\cite{schwartzyeh1982}. Broader discussions appear in \cite{Dufresne}, while
The tails of sums and differences of lognormal variables are studied in
\cite{archil-peter}. Financial applications include arithmetic-average
approximations \cite{turnbullwakeman1991,levy1992}, reciprocal-gamma methods
\cite{milevskyposner1998}, and higher-moment constructions
\cite{deelstra2010moment,leccadito2016,wu2019,hu2023pricingbasketoptionsmoments}.
These studies suggest that skewness and kurtosis can materially affect option
values, especially when maturity, volatility, or heterogeneous weights make
The basket distribution is more asymmetric.

The present paper follows a probability-first strategy. We do not approximate
the payoff at the outset. Instead, a change of measure first rewrites the exact
basket price as a linear combination of probabilities. For a standard basket,
defined here as a basket with one positive weight and all remaining weights
negative, these probabilities become CDF values of positive correlated
lognormal sums. This transformation separates the exact pricing identity from
the approximation used to evaluate it. Each positive sum is then represented
by a shifted lognormal variance mixture that matches the first four moments and
retains an explicit CDF.

A general mixed-sign basket requires a different proxy because its value may
fall on either side of zero. We therefore use a signed shifted lognormal
variable. Its sign probability, scale, location, and shift are selected to
match the first four moments of the basket. The resulting call price remains
analytical. Because the nonlinear moment equations need not produce a valid
distribution, we state explicit admissibility conditions and a numerical
root-selection rule. We also show that an admissible direct proxy generates
nonnegative, decreasing, and convex call prices and satisfies put--call parity.

The pricing-error analysis clarifies what moment matching can and cannot
provide. For a general basket, the option-price error is the discounted
integral of the CDF difference above the strike. For a standard basket, the
error is a weighted sum of CDF errors evaluated at a common threshold. These
identities show that four matched moments do not, by themselves, guarantee an
accurate option value. The approximation must also reproduce the part of the
distribution that contributes to exercise.

The numerical study preserves the original six standard-basket examples and
then introduces an empirical energy application. The normalized $3{:}2{:}1$
crack spread combines two units of RBOB gasoline, one unit of ULSD, and three
short units of WTI crude oil on a per-barrel basis \cite{cme2024crack}. Its
normalized weight vector has two positive entries and one negative entry, so it
provides a natural application of the signed general-basket method. Daily Yahoo
Finance futures data are used to estimate volatility and correlation inputs
\cite{yahoo2026energy}. Monte Carlo simulation under the same fitted lognormal
model is used as the benchmark. The empirical exercise therefore measures the
approximation error of the analytical formulas; it is not a calibration to
observed option prices.

The paper makes four main contributions. First, it derives an exact
probability representation for a general-basket option and converts the
standard-basket probabilities into CDFs of positive correlated lognormal sums.
Second, it develops a four-moment shifted-mixture approximation with an
explicit CDF and extends the same moment-matching idea to unrestricted basket
weights through a signed proxy. Third, it gives admissibility, root-selection,
and financial consistency results for the resulting analytical prices.
Fourth, it links option-price errors to CDF and tail discrepancies and verifies
this connection in both standard-basket diagnostics and an empirical
mixed-sign basket.

We begin under the multivariate Black--Scholes model,
\[
\ud S_i(t)=rS_i(t)\,\ud t+\sigma_iS_i(t)\,\ud W_i(t),
\qquad
\ud\langle W_i,W_j\rangle_t=\rho_{ij}\,\ud t.
\]
The time-zero value of a European basket call is
\[
C(T)=e^{-rT}\E\!\left[
\left(\sum_{i=1}^{n}\omega_iS_i(T)-K\right)^+
\right],
\]
where the weights may have either sign. Section~\ref{sec:standard-basket-framework}
develops the standard-basket probability representation.
Section~\ref{sec:general-basket-four-moment} presents the signed general-basket
proxy. Section~\ref{sec:admissibility-properties} studies admissibility,
financial properties, and pricing-error identities. Section~\ref{sec:numerical-evaluation}
reports the standard-basket benchmark and the empirical crack-spread
application.

\section{Probability Representations for Standard Basket Options}
\label{sec:standard-basket-framework}

This section develops the standard-basket pricing method in three steps. We
first use a change of measure to write a general basket price as a linear
combination of Gaussian probabilities. We then specialize the representation
to a standard basket and express each probability as the CDF of a positive sum
of correlated lognormal variables. These two steps are exact. The approximation
enters only in the final step, where each lognormal sum is replaced by a
four-moment shifted lognormal variance mixture. This ordering makes clear which
parts of the method are analytical identities and which parts are
approximations.

\subsection{A change-of-measure representation for general baskets}
\label{subsec:change-of-measure}

We begin with two Gaussian shift identities. They provide the
change-of-measure argument used in the basket-price representation.

\begin{lemma}[Independent Gaussian shift]
\label{lem:independent-gaussian-shift}
Let $\epsilon_1,\ldots,\epsilon_n$ be independent standard normal random
variables. For constants $a_1,\ldots,a_n$ and any measurable function $f$ for
which the expectations are finite,
\begin{equation}
\label{eq:independent-gaussian-shift}
\E\!\left[
\exp\!\left(\sum_{i=1}^{n}\left(a_i\epsilon_i-\frac12a_i^2\right)\right)
f(\epsilon_1,\ldots,\epsilon_n)
\right]
=
\E\!\left[f(\epsilon_1+a_1,\ldots,\epsilon_n+a_n)\right].
\end{equation}
\end{lemma}

\begin{proof}
Let $\varphi_n$ denote the density of an $n$-dimensional standard normal
vector. Completing the square gives
\[
\exp\!\left(\sum_{i=1}^{n}\left(a_ix_i-\frac12a_i^2\right)\right)
\varphi_n(x)
=
\varphi_n(x-a),
\]
where $a=(a_1,\ldots,a_n)^\top$. Substituting this identity into the integral
on the left-hand side of \eqref{eq:independent-gaussian-shift} and changing
variables from $x$ to $x+a$ proves the result.
\end{proof}

\begin{lemma}[Correlated Gaussian shift]
\label{lem:correlated-gaussian-shift}
Let $\epsilon=(\epsilon_1,\ldots,\epsilon_n)^\top$ be a centered Gaussian
vector with correlation matrix $\Gamma=(\Gamma_{ij})$. For any
$i\in\{1,\ldots,n\}$ and any measurable function $f$ for which the
expectations are finite,
\begin{equation}
\label{eq:correlated-gaussian-shift}
\E\!\left[e^{a_i\epsilon_i-\frac12a_i^2}
f(\epsilon_1,\ldots,\epsilon_n)\right]
=
\E\!\left[
f\!\left(\epsilon_1+a_i\Gamma_{1i},\ldots,
\epsilon_n+a_i\Gamma_{ni}\right)
\right].
\end{equation}
\end{lemma}

\begin{proof}
Write $\epsilon=A\widetilde\epsilon$, where $AA^\top=\Gamma$ and
$\widetilde\epsilon$ is a vector of independent standard normal variables.
If $A_i$ denotes the $i$th row of $A$, then
$\epsilon_i=A_i\widetilde\epsilon$ and $\lVert A_i\rVert^2=1$. Apply
Lemma~\ref{lem:independent-gaussian-shift} to $\widetilde\epsilon$ with shift
vector $a_iA_i^\top$. The $j$th transformed component is
\[
A_j\left(\widetilde\epsilon+a_iA_i^\top\right)
=
\epsilon_j+a_iA_jA_i^\top
=
\epsilon_j+a_i\Gamma_{ji},
\]
which yields \eqref{eq:correlated-gaussian-shift}.
\end{proof}

The preceding identities give the following exact basket-price
representation.

\begin{proposition}[Probability representation for a general basket]
\label{prop1}
The basket call price satisfies
\begin{equation}
\label{eq1}
\begin{aligned}
C(T)
={}&\sum_{i=1}^{n}\omega_iS_i(0)
\PP\!\left(
\sum_{j=1}^{n}\omega_jS_j(0)
 e^{-\frac12\sigma_j^2T+\rho_{ij}\sigma_i\sigma_jT
 +\sigma_j\sqrt{T}\epsilon_j}
-Ke^{-rT}\geq0
\right)\\
&-Ke^{-rT}
\PP\!\left(
\sum_{j=1}^{n}\omega_jS_j(0)
 e^{-\frac12\sigma_j^2T+\sigma_j\sqrt{T}\epsilon_j}
-Ke^{-rT}\geq0
\right).
\end{aligned}
\end{equation}
\end{proposition}

\begin{proof}
Let $A_T=\{\sum_{j=1}^{n}\omega_jS_j(T)\geq K\}$. Expanding the payoff gives
\[
C(T)
=
\sum_{i=1}^{n}\omega_i e^{-rT}\E[S_i(T)\1_{A_T}]
-Ke^{-rT}\PP(A_T).
\]
For each asset term, factor the lognormal martingale
$e^{\sigma_i\sqrt{T}\epsilon_i-\frac12\sigma_i^2T}$ and apply
Lemma~\ref{lem:correlated-gaussian-shift} with $a_i=\sigma_i\sqrt{T}$.
The probability term requires no change of measure. This gives
\eqref{eq1}.
\end{proof}

\begin{remark}[Gaussian-region interpretation]
\label{rem:gaussian-region-interpretation}
For $i=1,\ldots,n$, define
\[
D_i=
\left\{x\in\R^n:
\sum_{j=1}^{n}\omega_jS_j(0)
 e^{-\frac12\sigma_j^2T+\rho_{ij}\sigma_i\sigma_jT
 +\sigma_j\sqrt{T}x_j}
-Ke^{-rT}\geq0
\right\},
\]
and define
\[
D=
\left\{x\in\R^n:
\sum_{j=1}^{n}\omega_jS_j(0)
 e^{-\frac12\sigma_j^2T+\sigma_j\sqrt{T}x_j}
-Ke^{-rT}\geq0
\right\}.
\]
If $\phi_\Gamma$ is the density of an $n$-dimensional centered normal vector
with correlation matrix $\Gamma$, then Proposition~\ref{prop1} can be written
as
\[
C(T)=
\sum_{i=1}^{n}\omega_iS_i(0)\int_{D_i}\phi_\Gamma(x)\,\ud x
-Ke^{-rT}\int_D\phi_\Gamma(x)\,\ud x.
\]
Using complementary regions gives the put--call parity relation
\[
C(T)-P(T)=\sum_{i=1}^{n}\omega_iS_i(0)-Ke^{-rT},
\]
where $P(T)$ is the corresponding basket put price. The representation is
exact, but the regions $D_i$ and $D$ have nonlinear boundaries. Their Gaussian
probabilities therefore still require numerical integration or an additional
approximation.
\end{remark}

\subsection{Sums of correlated lognormal variables}
\label{subsec:standard-lognormal-sums}

We now specialize the general representation to a standard basket. Without
loss of generality, assume that $\omega_1>0$ and $\omega_i<0$ for
$2\leq i\leq n$, and define
\[
s_i=|\omega_i|S_i(0),\qquad i=1,\ldots,n.
\]
Set $q_1=1$, $q_i=-1$ for $2\leq i\leq n+1$, and
$s_{n+1}=Ke^{-rT}$. For notational convenience, also set
$\rho_{n+1,1}=0$ and $\sigma_{n+1}=1$. Proposition~\ref{prop1} then gives the
compact exact representation
\begin{equation}
\label{compact}
C^{sd}(T)=
\sum_{i=1}^{n+1}q_is_i
\PP\!\left(
\sum_{j=1}^{n}q_js_j
 e^{-\frac12\sigma_j^2T+\rho_{ij}\sigma_i\sigma_jT
 +\sigma_j\sqrt{T}\epsilon_j}
-Ke^{-rT}\geq0
\right).
\end{equation}
The probabilities in \eqref{compact} remain difficult to evaluate because their
exercise regions have nonlinear boundaries. The next proposition rewrites them
as CDF values of positive correlated lognormal sums.

\begin{proposition}\label{cor2.8}
The price of a standard-basket option admits the representation
\begin{equation} \label{SB2}
\begin{split}
C^{sd}(T)=&s_1P\left(\sum_{j=1}^ne^{N_j(m_{1j}, \bar{\sigma}_j^2)}\le 1\right)\\
-&\sum_{i=2}^ns_iP\left(\sum_{j=1}^ne^{N_j(m_{ij}, \bar{\sigma}_j^2)}\le 1\right)\\
-&Ke^{-rT}P\left(\sum_{j=1}^ne^{N_j(m_{n+1,j}, \bar{\sigma}_j^2)}\le 1\right),
\end{split}
\end{equation}
where
\begin{equation}
\begin{split}
\bar{\sigma}_1=&\sigma_1\sqrt{T}\;\;\;\;  \bar{\sigma}_j=\sqrt{\sigma_j^2+\sigma_1^2-2\sigma_1\sigma_j\rho_{1j}}\sqrt{T}, \; \;\; \;  2\le j\le n,\\
m_{i1}=&\ln (Ke^{-rT})-\ln s_1+\frac{1}{2}\sigma_1^2T-\rho_{i1}\sigma_i\sigma_1T, \;\;\;\; 1\le i\le n,\\
m_{ij}=&\ln s_j-\ln s_1+\frac{1}{2}(\sigma_1^2-\sigma_j^2)T+\rho_{ij}\sigma_i\sigma_jT-\rho_{i1}\sigma_i\sigma_1T, \;\;\;\; 1\le i\le n, \; 2\le j\le n,\\
m_{n+1,1}=&\ln(Ke^{-rT})-\ln s_1+\frac{1}{2}\sigma_1^2T,\\
m_{n+1,j}=&\ln s_j-\ln s_1+\frac{1}{2}(\sigma_1^2-\sigma_j^2)T,\;\;\;\;2\le j\le n,
\end{split}
\end{equation}
and the correlation matrix $\bar{\Gamma}=(\theta_{jk})$ of the normal random variables $(N_1, N_2, \ldots, N_n)$ is given by
\begin{equation}\label{kl}
\begin{split}
\theta_{11}=&1, \; \; \; \theta_{1j}=\frac{\sigma_1-\sigma_j\rho_{1j}}{\sqrt{\sigma_j^2+\sigma_1^2-2\sigma_1\sigma_j\rho_{1j}}},\; \;\; \;  2\le j\le n\\
\theta_{j, k}=&\frac{\sigma_j\sigma_k\rho_{jk}-\sigma_1\sigma_j\rho_{1j}-\sigma_1\sigma_k\rho_{1k}+\sigma_1^2}{\sqrt{\sigma_j^2+\sigma_1^2-2\sigma_1\sigma_j\rho_{1j}}\sqrt{\sigma_k^2+\sigma_1^2-2\sigma_1\sigma_k\rho_{1k}}}, \;\;\;\;\; 2\le j\le n, \; \; 2\le k\le n\\
\end{split}
\end{equation}
\end{proposition}
\begin{proof}
Using \eqref{compact}, divide each event by its positive first-asset term. For
$i=1,\ldots,n+1$, this gives
\begin{equation}
\label{eq:standard-basket-transformed-event}
C^{sd}(T)=
\sum_{i=1}^{n+1}q_is_i
\PP\!\left(
 e^{m_{i1}-\sigma_1\sqrt{T}\epsilon_1}
 +\sum_{j=2}^{n}
 e^{m_{ij}+(\sigma_j\epsilon_j-\sigma_1\epsilon_1)\sqrt{T}}
 \leq1
\right),
\end{equation}
where
\[
m_{i1}=\ln(Ke^{-rT})-\ln s_1+\frac12\sigma_1^2T
-\rho_{i1}\sigma_i\sigma_1T
\]
and, for $j\geq2$,
\[
m_{ij}=\ln s_j-\ln s_1+\frac12(\sigma_1^2-\sigma_j^2)T
+\rho_{ij}\sigma_i\sigma_jT
-\rho_{i1}\sigma_i\sigma_1T.
\]
For $j\geq2$, define
$w_j=\sigma_j\epsilon_j-\sigma_1\epsilon_1$ and set
$N_1=-\epsilon_1$. Then
\[
w_j=\sqrt{\sigma_j^2+\sigma_1^2
-2\sigma_1\sigma_j\rho_{1j}}\,N_j,
\]
where each $N_j$ is standard normal. Direct covariance calculations give
\[
\Corr(N_1,N_j)=
\frac{\sigma_1-\sigma_j\rho_{1j}}
{\sqrt{\sigma_j^2+\sigma_1^2-2\sigma_1\sigma_j\rho_{1j}}}
\]
and, for $j,k\geq2$,
\[
\Corr(N_j,N_k)=
\frac{\sigma_j\sigma_k\rho_{jk}
-\sigma_1\sigma_j\rho_{1j}
-\sigma_1\sigma_k\rho_{1k}+\sigma_1^2}
{\sqrt{\sigma_j^2+\sigma_1^2-2\sigma_1\sigma_j\rho_{1j}}
 \sqrt{\sigma_k^2+\sigma_1^2-2\sigma_1\sigma_k\rho_{1k}}}.
\]
Substituting these standardized Gaussian variables into
\eqref{eq:standard-basket-transformed-event} yields \eqref{SB2} and the stated
parameters.
\end{proof}

Proposition~\ref{cor2.8} transfers the main difficulty from a nonlinear exercise region to the distribution of a positive sum of correlated lognormal variables. This change is useful because the same distributional approximation can be applied to each probability in \eqref{SB2}.

\begin{remark}
The representation in \eqref{SB2} is exact. Its practical value depends on the quality of the approximation used for the lognormal sums. Three-moment shifted lognormal methods are accurate in many regular-basket settings \cite{hu2023pricingbasketoptionsmoments,wu2019}. The construction below adds a fourth-moment condition so that the proxy can also adjust to the kurtosis of the target distribution.
\end{remark}

\subsection{A four-moment approximation}
\label{subsec:four-moment-lognormal}

We now construct an analytical approximation for the sums that appear in Proposition~\ref{cor2.8}. A finite sum of lognormal variables is generally not lognormal, even when the underlying Gaussian variables are jointly normal. The classical Fenton--Wilkinson approximation replaces the sum with a single lognormal variable matched to the first two moments. This method is simple, but two moments cannot fully represent the skewness and tail shape of a heterogeneous or strongly dependent sum. Shifted-lognormal and recursive approximations provide additional flexibility; see, for example, \cite{schwartzyeh1982,Dufresne,hu2023pricingbasketoptionsmoments}. Our objective is to retain an explicit cumulative distribution function while matching the first four moments of the target sum.

We use the shifted mixture
\[
W=e^{s\sqrt{Y}N+m}+\tau,
\]
where $N$ is standard normal, $Y$ is independent of $N$, and
\[
P(Y=\alpha)=p,\qquad P(Y=\beta)=1-p.
\]
The constants $\alpha$ and $\beta$ are fixed positive values in the baseline specification, while $p$, $s$, $m$, and $\tau$ are determined by moment matching. Conditional on $Y$, the random variable $W-\tau$ is lognormal. The binary variance mixture gives the approximation enough flexibility to match skewness and kurtosis without losing an explicit distribution function. Section~\ref{subsec:tail-sensitive-mixture} also discusses the symmetric parameterization $\alpha=1-\delta$ and $\beta=1+\delta$, which provides a one-dimensional way to examine the sensitivity of the approximation to tail shape.

Let
\[
S = \sum_{j=1}^d e^{N(m_j,\nu_j^2)}
\]
be the sum of lognormal terms, and define
\[
\bar{S} = \sum_{j=1}^d p_j e^{\nu_j N_j},
\]
where
\[
p_j := \frac{e^{m_j}}{\sum_{i=1}^d e^{m_i}}, \quad j = 1,2,\ldots,d.
\]
Then we can rewrite
\[
S = \Big(\sum_{j=1}^d e^{m_j}\Big)\bar{S}.
\]
Our objective is to approximate $\bar{S}$ by $W$ via matching the first four moments, that is, we impose
\[
E[\bar{S}^i] = E[W^i], \quad i = 1,2,3,4.
\]
This system determines the parameters $\tau$, $s$, $m$, and $p$. Consequently, we approximate
\begin{equation}\label{log-20}
P(S \le x) = P\!\Big(\bar{S} \le \frac{x}{\sum_{i=1}^d e^{m_i}}\Big)
\approx P\!\Big(W \le \frac{x}{\sum_{i=1}^d e^{m_i}}\Big).
\end{equation}

The first four moments of $\bar{S}$ are given by
\begin{equation}
\begin{split}
 E[\bar{S}]&=\sum_{j=1}^dp_je^{\frac{1}{2}\nu_j^2},\\
 E[\bar{S}^2]&=\sum_{i=1}^d\sum_{j=1}^dp_ip_je^{\frac{1}{2}\nu_i^2+\frac{1}{2}\nu_j^2+\rho_{ij}\nu_i\nu_j},\\
 E[\bar{S}^3]&=\sum_{i=1}^d\sum_{j=1}^d\sum_{k=1}^dp_ip_jp_ke^{\frac{1}{2}\nu_i^2+\frac{1}{2}\nu_j^2+\frac{1}{2}\nu_k^2+\rho_{ij}\nu_i\nu_j+\rho_{ik}\nu_i\nu_k+\rho_{jk}\nu_j\nu_k},\\
 E[\bar{S}^4]&=\sum_{i=1}^d\sum_{j=1}^d\sum_{k=1}^d\sum_{l=1}^dp_ip_jp_kp_le^{\frac{1}{2}\nu_i^2+\frac{1}{2}\nu_j^2+\frac{1}{2}\nu_k^2+\frac{1}{2}\nu_l^2+\rho_{ij}\nu_i\nu_j+\rho_{ik}\nu_i\nu_k+\rho_{il}\nu_i\nu_l+\rho_{jk}\nu_j\nu_k+\rho_{jl}\nu_j\nu_l+\rho_{kl}\nu_k\nu_l}.
\end{split}
\end{equation}
The first four moments of $W$ are
\begin{equation}\label{W_moments}
\begin{split}
 E[W]&=e^m\phi_Y(\frac{1}{2}s^2)+\tau,\\
 E[W^2]&=e^{2m}\phi_Y(2s^2)+2\tau e^m\phi_Y(\frac{1}{2}s^2)+\tau^2,\\
 E[W^3]&=e^{3m}\phi_Y(\frac{9}{2}s^2)+3\tau e^{2m}\phi_Y(2s^2)+3\tau^2 e^m\phi_Y(\frac{1}{2}s^2)+\tau^3,\\
 E[W^4]&=e^{4m}\phi_Y(8s^2)+4\tau e^{3m}\phi_Y(\frac{9}{2}s^2)+6\tau^2 e^{2m}\phi_Y(2s^2)+4\tau^3 e^m\phi_Y(\frac{1}{2}s^2)+\tau^4,
\end{split}
\end{equation}
where
\begin{equation*}
\begin{split}
 \phi_Y(\theta)=E[e^{\theta Y}]=pe^{\alpha\theta}+(1-p)e^{\beta\theta}.
\end{split}
\end{equation*}

The mixture structure gives an explicit cumulative distribution function.

\begin{lemma}\label{lemma31} For $x > \tau$, the distribution function of $W$ is
\begin{equation*}
\begin{split}
P(W \le x) \;=\; p\,\Phi\!\left(\frac{\ln(x-\tau)-m}{s\sqrt{\alpha}}\right)
\;+\; (1-p)\,\Phi\!\left(\frac{\ln(x-\tau)-m}{s\sqrt{\beta}}\right).
\end{split}
\end{equation*}
For $x \leq \tau$, we instead obtain $P(W \leq x) = 0$.
\end{lemma}
\begin{proof}
If $x\leq\tau$, then $W>\tau$ almost surely. For $x>\tau$, using the relation $W=e^{s\sqrt{Y}N+m}+\tau$, we obtain
\begin{equation*}
\begin{split}
 P(W\le x)&=P(e^{s\sqrt{Y}N+m}+\tau\le x)=P(N\le\frac{\ln(x-\tau)-m}{s\sqrt{Y}})\\
 &=pP(N\le\frac{\ln(x-\tau)-m}{s\sqrt{\alpha}})+(1-p)P(N\le\frac{\ln(x-\tau)-m}{s\sqrt{\beta}})\\
 &=p\Phi(\frac{\ln(x-\tau)-m}{s\sqrt{\alpha}})+(1-p)\Phi(\frac{\ln(x-\tau)-m}{s\sqrt{\beta}}).
\end{split}
\end{equation*}
\end{proof}

The parameters are determined by the four matching conditions
\begin{equation}\label{match}
\begin{split}
 E[W]=E[\bar{S}],\quad E[W^2]=E[\bar{S}^2],\quad E[W^3]=E[\bar{S}^3],\quad E[W^4]=E[\bar{S}^4].
\end{split}
\end{equation}
The next proposition reduces these conditions to two equations for $x=s^2$ and $p$, after which $m$ and $\tau$ follow directly.

\begin{proposition}\label{32}
Suppose that $(x,p)$ solves the system
\begin{equation*}
\begin{split}
 &\phi_Y(\frac{9}{2}x)-3\phi_Y(\frac{1}{2}x)\phi_Y(2x)+2(\phi_Y(\frac{1}{2}x))^3-\eta_{\bar{S}}(\phi_Y(2x)-(\phi_Y(\frac{1}{2}x))^2)^{\frac{3}{2}}=0,\\
 &\phi_Y(8x)-4\phi_Y(\frac{9}{2}x)\phi_Y(\frac{1}{2}x)+6\phi_Y(2x)(\phi_Y(\frac{1}{2}x))^2-3(\phi_Y(\frac{1}{2}x))^4\\
 -&\kappa_{\bar{S}}(\phi_Y(2x)-(\phi_Y(\frac{1}{2}x))^2)^2=0,
\end{split}
\end{equation*}
and let
\begin{equation*}
\begin{split}
 s&=\sqrt{x},\\
 m&=\frac{1}{2}\ln\left(\frac{\sigma_{\bar{S}}^2}{\phi_Y(2s^2)-(\phi_Y(\frac{1}{2}s^2))^2}\right),\\
 \tau&=\mu_{\bar{S}}-\sigma_{\bar{S}}\frac{\phi_Y(\frac{1}{2}s^2)}{\sqrt{\phi_Y(2s^2)-(\phi_Y(\frac{1}{2}s^2))^2}},
\end{split}
\end{equation*}
where $\mu_{\bar{S}},\sigma_{\bar{S}},\eta_{\bar{S}},\kappa_{\bar{S}}$ denote the mean, standard deviation, skewness, and kurtosis of $\bar{S}$, respectively. Then the parameters $p$, $s$, $m$, and $\tau$ satisfy the four matching conditions in \eqref{match}.
 \end{proposition}

\begin{proof}
Write
\[
A_1=\phi_Y\!\left(\frac12s^2\right),\qquad
A_2=\phi_Y(2s^2),\qquad
A_3=\phi_Y\!\left(\frac92s^2\right),\qquad
A_4=\phi_Y(8s^2).
\]
The mean and variance of $W$ are
\[
\E[W]=e^mA_1+\tau,
\qquad
\Var(W)=e^{2m}(A_2-A_1^2).
\]
A direct central-moment calculation gives
\[
\eta_W=
\frac{A_3-3A_1A_2+2A_1^3}{(A_2-A_1^2)^{3/2}}
\]
and
\[
\kappa_W=
\frac{A_4-4A_1A_3+6A_1^2A_2-3A_1^4}{(A_2-A_1^2)^2}.
\]
Equating these expressions to the skewness and kurtosis of $\bar S$ gives the
two equations in the proposition after setting $x=s^2$. The variance equation
then yields
\[
m=\frac12\ln\!\left(
\frac{\sigma_{\bar S}^2}{A_2-A_1^2}
\right),
\]
and the mean equation gives
\[
\tau=\mu_{\bar S}-
\sigma_{\bar S}\frac{A_1}{\sqrt{A_2-A_1^2}}.
\]
These values satisfy all four matching conditions.
\end{proof}

The combined procedure for a standard basket is now complete. Proposition~\ref{cor2.8} identifies the correlated lognormal sums that determine the exact price, Proposition~\ref{32} matches each normalized sum with the four-moment proxy, and Lemma~\ref{lemma31} evaluates the resulting probability in closed form. Section~\ref{sec:numerical-evaluation} applies this construction to three-asset standard-basket options. We next extend the four-moment idea to baskets with unrestricted signs.

\section{Approximation for General Basket Options}
\label{sec:general-basket-four-moment}

Previously, we used positive lognormal sums to price standard baskets. We now consider a general basket whose positive and negative weights may yield values on either side of zero. To accommodate this feature, we extend the shifted lognormal construction of \cite{hu2023pricingbasketoptionsmoments} by introducing an independent sign variable. The approximating random variable is
\[
\pi=A(e^{sN+m}+\tau),
\]
where $A$ is independent of $N$ and satisfies $P(A=1)=p$ and $P(A=-1)=1-p$, with $0<p<1$. Its first four moments, denoted by $M_k=E[\pi^k]$, are
\begin{equation}\label{4moments}
\begin{split}
 M_1&=(2p-1)(e^{\frac{1}{2}s^2+m}+\tau),\\
 M_2&=e^{2s^2+2m}+2\tau e^{\frac{1}{2}s^2+m}+\tau^2,\\
 M_3&=(2p-1)(e^{\frac{9}{2}s^2+3m}+3\tau e^{2s^2+2m}+3\tau^2e^{\frac{1}{2}s^2+m}+\tau^3),\\
 M_4&=e^{8s^2+4m}+4\tau e^{\frac{9}{2}s^2+3m}+6\tau^2e^{2s^2+2m}+4\tau^3e^{\frac{1}{2}s^2+m}+\tau^4.
\end{split}
\end{equation}

Suppose the basket is given by
\begin{equation}\label{basket}
\begin{split}
 B(T)=\sum_{i=1}^{n}w_{i}S_{i}(0)e^{(r-\frac{1}{2}\sigma_{i}^{2})T+\sigma_{i}\sqrt{T}N_{i}}.
\end{split}
\end{equation}
Our goal is to match the first four moments of $B(T)$ with those of $\pi$ and recover the parameters $p$, $s$, $m$, and $\tau$. The next lemma gives the required basket moments.

\begin{lemma}[Moments of a correlated lognormal basket]\label{lem:general-basket-moments}
Suppose the basket is given by \eqref{basket} and let $\{\rho_{ij}\}$ denote the corresponding correlation coefficients. Then we have
\begin{equation}
\begin{split}
 E[B(T)]&=e^{rT}\sum_{i=1}^{n}w_{i}S_{i}(0)\\
 E[(B(T))^{2}]&=e^{2rT}\sum_{i=1}^{n}\sum_{j=1}^{n}w_{i}w_{j}S_{i}(0)S_{j}(0)e^{\rho_{ij}\sigma_{i}\sigma_{j}T}\\
 E[(B(T))^{3}]&=e^{3rT}\sum_{i=1}^{n}\sum_{j=1}^{n}\sum_{k=1}^{n}w_{i}w_{j}w_{k}S_{i}(0)S_{j}(0)S_{k}(0)e^{(\rho_{ij}\sigma_{i}\sigma_{j}+\rho_{ik}\sigma_{i}\sigma_{k}+\rho_{jk}\sigma_{j}\sigma_{k})T}\\
 E[(B(T))^{4}]&=e^{4rT}\sum_{i=1}^{n}\sum_{j=1}^{n}\sum_{k=1}^{n}\sum_{l=1}^{n}w_{i}w_{j}w_{k}w_{l}S_{i}(0)S_{j}(0)S_{k}(0)S_{l}(0)\\
 &\times e^{(\rho_{ij}\sigma_{i}\sigma_{j}+\rho_{ik}\sigma_{i}\sigma_{k}+\rho_{il}\sigma_{i}\sigma_{l}+\rho_{jk}\sigma_{j}\sigma_{k}+\rho_{jl}\sigma_{j}\sigma_{l}+\rho_{kl}\sigma_{k}\sigma_{l})T}.
\end{split}
\end{equation}
\end{lemma}
\begin{proof}
Write
\[
X_i=w_iS_i(0)\exp\!\left(\left(r-\frac12\sigma_i^2\right)T
+\sigma_i\sqrt{T}N_i\right),
\qquad
B(T)=\sum_{i=1}^{n}X_i.
\]
For any indices $i_1,\ldots,i_k$, joint normality gives
\[
E\!\left[\prod_{a=1}^{k}X_{i_a}\right]
=
e^{krT}\prod_{a=1}^{k}w_{i_a}S_{i_a}(0)
\exp\!\left(
\sum_{1\leq a<b\leq k}
\rho_{i_ai_b}\sigma_{i_a}\sigma_{i_b}T
\right).
\]
The terms $-\frac12\sigma_{i_a}^2T$ in the lognormal drifts cancel the
diagonal terms produced by the Gaussian moment-generating function.
Expanding $B(T)^k$ and applying this identity for $k=1,2,3,4$ gives the four
formulas.
\end{proof}

\begin{remark}[Futures-price implementation]
\label{rem:futures-implementation}
The empirical application in Section~\ref{subsec:crack-application} is based on
futures prices rather than spot prices. Under the risk-neutral measure, a
futures price with constant volatility is modeled as
\[
F_i(T)=F_i(0)\exp\!\left(
-\frac12\sigma_i^2T+\sigma_i\sqrt{T}N_i
\right).
\]
The moment formulas above continue to apply after removing the common factors
$e^{krT}$ from the $k$th moments. Discounting remains outside the payoff
expectation. Thus, the futures implementation changes only the drift
convention used to compute the basket moments; the signed-proxy construction
and the option-pricing formula are unchanged.
\end{remark}

The four parameters of the signed proxy are recovered in two stages. For a
fixed sign probability $p$, the first three moments determine a shifted
lognormal component. The fourth moment then determines $p$.

\begin{theorem}[Four-moment parameter recovery]
\label{th12}
Let $b_k=\E[B(T)^k]$ for $k=1,2,3,4$, and suppose that the skewness of
$B(T)$ is nonzero. Search over $p\in(1/2,1)$ when the basket skewness is
positive and over $p\in(0,1/2)$ when it is negative. For each candidate $p$,
define
\[
\bar M_1(p)=\frac{b_1}{2p-1},
\qquad
\bar M_2(p)=b_2,
\qquad
\bar M_3(p)=\frac{b_3}{2p-1},
\]
and
\begin{equation}
\label{eq:conditional-moments}
\begin{aligned}
\mu(p)&=\bar M_1(p),\\
\sigma(p)&=\sqrt{\bar M_2(p)-\bar M_1(p)^2},\\
\eta(p)&=
\frac{\bar M_3(p)-3\bar M_2(p)\bar M_1(p)+2\bar M_1(p)^3}
{\sigma(p)^3}.
\end{aligned}
\end{equation}
Set
\begin{equation}
\label{eq:stable-xp}
x(p)=2\cosh\!\left[
\frac{2}{3}\operatorname{arsinh}\!\left(\frac{|\eta(p)|}{2}\right)
\right]-1,
\end{equation}
and define
\begin{equation}
\label{eq:signed-proxy-parameters}
\begin{aligned}
s(p)&=\sqrt{\ln x(p)},\\
m(p)&=\frac12\ln\!\left(
\frac{\sigma(p)^2}{x(p)[x(p)-1]}
\right),\\
\tau(p)&=\mu(p)-\frac{\sigma(p)}{\sqrt{x(p)-1}}.
\end{aligned}
\end{equation}
If $p$ is an admissible solution of
\begin{equation}
\label{p}
\begin{aligned}
0={}&b_4-
\Bigl\{
\sigma(p)^4\bigl[x(p)^4+2x(p)^3+3x(p)^2-3\bigr]\\
&\quad+4\mu(p)\sigma(p)^3\bigl[x(p)+2\bigr]\sqrt{x(p)-1}
+6\mu(p)^2\sigma(p)^2+\mu(p)^4
\Bigr\},
\end{aligned}
\end{equation}
then $p$, $s(p)$, $m(p)$, and $\tau(p)$ match the first four moments of
$B(T)$ with those of the signed proxy
$\pi=A(e^{sN+m}+\tau)$.
\end{theorem}

\begin{remark}[Stable evaluation of $x(p)$]
\label{rem:stable-x}
The hyperbolic expression in \eqref{eq:stable-xp} is algebraically equivalent
to the usual Cardano formula. It is preferable in numerical work because it
remains real valued when floating-point cancellation makes one cube-root
argument slightly negative. This issue is relevant when the fitted probability
is close to $0$ or $1$.
\end{remark}

\begin{proof}
For a fixed $p\neq1/2$, matching the first three moments of
$\pi=A(e^{sN+m}+\tau)$ is equivalent to matching the first three moments of the
positive component $e^{sN+m}+\tau$ with
$\bar M_1(p)$, $\bar M_2(p)$, and $\bar M_3(p)$. The standard three-moment
shifted lognormal calculation gives \eqref{eq:stable-xp} and
\eqref{eq:signed-proxy-parameters}. Substituting these parameters into the
fourth raw moment of the positive component yields
\[
\sigma(p)^4\bigl[x(p)^4+2x(p)^3+3x(p)^2-3\bigr]
+4\mu(p)\sigma(p)^3\bigl[x(p)+2\bigr]\sqrt{x(p)-1}
+6\mu(p)^2\sigma(p)^2+\mu(p)^4.
\]
The fourth moment of the sign variable is one, so the same expression is the
fourth moment of $\pi$. Equating it to $b_4$ gives \eqref{p}. The sign
restriction on $p-1/2$ ensures that the proxy and the basket have the same sign
of skewness. The admissibility conditions in
Definition~\ref{def:admissible-signed} ensure that all recovered parameters are
real and finite.
\end{proof}

\begin{theorem}[Analytical general-basket call approximation]
\label{thm:general-basket-price}
Let $p$, $s$, $m$, and $\tau$ be an admissible parameter set from
Theorem~\ref{th12}. Define
\begin{equation}
\label{eq:positive-branch-stoploss}
G_+(K)=
\begin{cases}
 e^{m+\frac12s^2}+\tau-K, & K\leq\tau,\\[0.3em]
 e^{m+\frac12s^2}\Phi(d_{11})-(K-\tau)\Phi(d_{12}), & K>\tau,
\end{cases}
\end{equation}
where, for $K>\tau$,
\[
d_{11}=\frac{m+s^2-\ln(K-\tau)}{s},
\qquad
d_{12}=\frac{m-\ln(K-\tau)}{s}.
\]
Also define
\begin{equation}
\label{eq:negative-branch-stoploss}
G_-(K)=
\begin{cases}
 (-K-\tau)\Phi(d_{22})-e^{m+\frac12s^2}\Phi(d_{21}), & K<-\tau,\\[0.3em]
 0, & K\geq-\tau,
\end{cases}
\end{equation}
where, for $K<-\tau$,
\[
d_{21}=\frac{\ln(-K-\tau)-m-s^2}{s},
\qquad
d_{22}=\frac{\ln(-K-\tau)-m}{s}.
\]
Then the general-basket call price is approximated by
\begin{equation}
\label{pp3_1}
\widehat\Pi(K)=e^{-rT}\bigl[pG_+(K)+(1-p)G_-(K)\bigr].
\end{equation}
\end{theorem}

\begin{proof}
Conditioning on the independent sign variable gives
\[
\E[(\pi-K)^+]
=p\E[(e^{sN+m}+\tau-K)^+]
+(1-p)\E[(-e^{sN+m}-\tau-K)^+].
\]
The first conditional payoff is always positive when $K\leq\tau$ and otherwise
is a standard truncated lognormal stop-loss expectation. The second conditional
payoff is zero when $K\geq-\tau$ and is a lower-tail truncated lognormal
expectation when $K<-\tau$. Evaluating these two expectations gives
\eqref{eq:positive-branch-stoploss} and
\eqref{eq:negative-branch-stoploss}; discounting gives \eqref{pp3_1}.
\end{proof}

\section{Admissibility, Financial Properties, and Pricing-Error Analysis}
\label{sec:admissibility-properties}

The moment equations derived in the previous sections are nonlinear. A numerical solution is useful only if it produces a valid probability model and finite pricing parameters. This section makes these requirements explicit. We first state the admissibility conditions for the shifted mixture used for positive lognormal sums and for the signed proxy used for general baskets. We then give a practical rule for selecting among several numerical roots and establish the main financial properties of the direct general-basket approximation. The last part of the section connects distributional approximation to option valuation by deriving exact pricing-error representations for both general and standard baskets.

\subsection{Admissibility and root-selection}
\label{subsec:admissibility-positive-sum}

Consider the proxy
\[
W=e^{s\sqrt{Y}N+m}+\tau,
\qquad
P(Y=\alpha)=p,\qquad P(Y=\beta)=1-p,
\]
introduced in Section~\ref{sec:standard-basket-framework}. Write $x=s^2$ and define
\begin{equation}\label{eq:mixture-variance-factor}
D_Y(x,p)
=
\phi_Y(2x)-\phi_Y\!\left(\frac{x}{2}\right)^2,
\end{equation}
where
\[
\phi_Y(u)=pe^{\alpha u}+(1-p)e^{\beta u}.
\]

\begin{definition}[Admissible lognormal-sum solution]
\label{def:admissible-positive-sum}
A pair $(x,p)$ solving the two equations in Proposition~\ref{32} is called admissible if
\[
\alpha>0,\qquad \beta>0,\qquad 0\leq p\leq 1,\qquad x>0,
\]
and
\[
D_Y(x,p)>0.
\]
In addition, the values of $m$ and $\tau$ recovered from Proposition~\ref{32} must be finite, and the four moment residuals must be smaller than a prescribed numerical tolerance.
\end{definition}

The condition $D_Y(x,p)>0$ ensures that the variance of $W$ is positive and that the logarithm appearing in the formula for $m$ is well defined. The endpoint cases $p=0$ and $p=1$ are allowed as degenerate limits. In those cases, the binary mixture reduces to a single shifted lognormal component. For any admissible solution, Lemma~\ref{lemma31} defines a genuine cumulative distribution function. In particular, for a threshold $z$,
\[
P(W\leq z)=0\qquad\text{when }z\leq\tau,
\]
while the two-component normal-CDF expression applies when $z>\tau$.

\begin{proposition}\label{prop:valid-mixture-cdf}
For every admissible solution in Definition~\ref{def:admissible-positive-sum}, the function
\[
F_W(z)=P(W\leq z)
\]
is nondecreasing, right-continuous, takes values in $[0,1]$, and satisfies
\[
\lim_{z\to-\infty}F_W(z)=0,
\qquad
\lim_{z\to\infty}F_W(z)=1.
\]
\end{proposition}

\begin{proof}
Conditional on $Y=\alpha$ or $Y=\beta$, the variable $W$ is a shifted lognormal random variable. Its conditional distribution function is therefore valid. The unconditional distribution function is a convex combination of the two conditional distribution functions, with weights $p$ and $1-p$. The stated properties follow immediately.
\end{proof}

For a general basket, the proxy is
\[
\pi=A\left(e^{sN+m}+\tau\right),
\qquad
P(A=1)=p,\qquad P(A=-1)=1-p.
\]
The formulas in Theorem~\ref{th12} contain the quantities
\[
\bar M_1(p)=\frac{E[B(T)]}{2p-1},
\qquad
\bar M_2(p)=E[B(T)^2],
\qquad
\bar M_3(p)=\frac{E[B(T)^3]}{2p-1}.
\]
Thus, $p=1/2$ is excluded from the direct parameter recovery whenever the odd moments are nonzero.

\begin{definition}[Admissible signed-proxy solution]
\label{def:admissible-signed}
A solution $p$ of \eqref{p} is called admissible if the following conditions hold:
\begin{enumerate}
\item $0\leq p\leq 1$ and $p\neq 1/2$;
\item $p>1/2$ when $\eta_{B(T)}>0$, and $p<1/2$ when $\eta_{B(T)}<0$;
\item
\[
\sigma(p)^2=\bar M_2(p)-\bar M_1(p)^2>0;
\]
\item $x(p)>1$, so that $s=\sqrt{\ln x(p)}$ is real and strictly positive;
\item the recovered values of $m$ and $\tau$ are finite;
\item the normalized residuals of the four matching equations are below a prescribed tolerance.
\end{enumerate}
\end{definition}

The cases $p=0$ and $p=1$ correspond to degenerate one-sided proxies and can be treated as limiting cases. By contrast, $p=1/2$ cannot be inserted into the formulas of Theorem~\ref{th12}. A basket with zero skewness therefore requires a separate symmetric treatment or a limiting numerical procedure. The present parameterization is intended for baskets with nonzero skewness.

\begin{remark}\label{rem:boundary-p}
Values of $p$ reported numerically as $0$ or $1$ may represent exact boundary solutions or interior solutions rounded to the displayed precision. In either case, the implementation should evaluate the unrounded value and record whether the fitted proxy is genuinely mixed or has effectively collapsed to one component.
\end{remark}

The nonlinear systems may have no admissible root, one admissible root, or several admissible roots. We use the following selection rule.

For a candidate parameter vector $\vartheta$, define the normalized moment residual
\begin{equation}\label{eq:moment-residual}
\mathcal R_4(\vartheta)
=
\sum_{k=1}^{4}
\frac{\left|E[Z_\vartheta^k]-E[Z^k]\right|}
     {1+\left|E[Z^k]\right|},
\end{equation}
where $Z$ denotes the target variable and $Z_\vartheta$ denotes its fitted proxy. For the positive-sum construction, $Z=\bar S$ and $Z_\vartheta=W$. For the general-basket construction, $Z=B(T)$ and $Z_\vartheta=\pi$.

The numerical procedure is:

\begin{enumerate}
\item search for all roots over the relevant parameter domain;
\item remove roots that fail Definitions~\ref{def:admissible-positive-sum} or~\ref{def:admissible-signed};
\item among the remaining roots, select the one with the smallest value of $\mathcal R_4$;
\item if several roots have nearly identical four-moment residuals, use the fifth-moment discrepancy as a diagnostic tie-breaker;
\item if no admissible four-moment solution is found, return the corresponding three-moment approximation and report that the fallback was used.
\end{enumerate}

This rule separates mathematical feasibility from numerical accuracy. It also prevents a solver from accepting a root that satisfies one transformed equation while producing an invalid variance, a complex parameter, or a poor reconstruction of the original moments.

\subsection{Financial properties of the direct general-basket approximation}
\label{subsec:financial-properties}

Let $\pi$ be an admissible signed proxy and define the approximate call and put prices by
\begin{equation}\label{eq:proxy-call-put}
\widehat C(K)=e^{-rT}E[(\pi-K)^+],
\qquad
\widehat P(K)=e^{-rT}E[(K-\pi)^+].
\end{equation}
Because these prices are expectations of standard option payoffs under a valid proxy distribution, they preserve the basic shape restrictions with respect to the strike.

\begin{proposition}[Strike monotonicity, convexity, and put--call parity]
\label{prop:financial-properties}
Suppose that the signed proxy is admissible and matches the first moment of $B(T)$. Then:

\begin{enumerate}
\item $\widehat C(K)\geq0$ and $\widehat P(K)\geq0$;
\item $\widehat C(K)$ is decreasing and convex in $K$, while $\widehat P(K)$ is increasing and convex in $K$;
\item for any $K_1,K_2\in\mathbb R$,
\[
\left|\widehat C(K_2)-\widehat C(K_1)\right|
\leq e^{-rT}|K_2-K_1|;
\]
\item the approximate prices satisfy
\begin{equation}\label{eq:proxy-put--call-parity}
\widehat C(K)-\widehat P(K)
=
e^{-rT}\left(E[B(T)]-K\right)
=
\sum_{i=1}^{n}w_iS_i(0)-Ke^{-rT};
\end{equation}
\item at every strike at which the distribution function of $\pi$ is continuous,
\begin{equation}\label{eq:proxy-strike-derivative}
\frac{\partial \widehat C(K)}{\partial K}
=
-e^{-rT}P(\pi>K).
\end{equation}
If $\pi$ has a density $f_\pi$, then
\begin{equation}\label{eq:proxy-second-strike-derivative}
\frac{\partial^2\widehat C(K)}{\partial K^2}
=
e^{-rT}f_\pi(K)\geq0.
\end{equation}
\end{enumerate}
\end{proposition}

\begin{proof}
Nonnegativity follows directly from the payoffs in \eqref{eq:proxy-call-put}. For every fixed value of $\pi$, the function $(\pi-K)^+$ is decreasing and convex in $K$, while $(K-\pi)^+$ is increasing and convex. Expectations preserve these properties. The Lipschitz bound follows from
\[
\left|(x-K_2)^+-(x-K_1)^+\right|\leq |K_2-K_1|.
\]
Moreover,
\[
(\pi-K)^+-(K-\pi)^+=\pi-K.
\]
Taking expectations, discounting, and using $E[\pi]=E[B(T)]$ gives \eqref{eq:proxy-put--call-parity}. Under the risk-neutral model,
\[
E[B(T)]=e^{rT}\sum_{i=1}^{n}w_iS_i(0),
\]
which gives the final expression. The derivative formulas are the standard strike derivatives of an option value written as an expectation under a continuous distribution.
\end{proof}

A further consequence of Jensen's inequality is the lower bound
\begin{equation}\label{eq:proxy-forward-lower-bound}
\widehat C(K)
\geq
e^{-rT}\left(E[B(T)]-K\right)^+
=
\left(\sum_{i=1}^{n}w_iS_i(0)-Ke^{-rT}\right)^+.
\end{equation}
Thus, the direct approximation respects the usual forward intrinsic-value bound.

\begin{remark}[Scope of the financial guarantees]
\label{rem:scope-financial-properties}
Proposition~\ref{prop:financial-properties} applies directly to the general-basket price obtained from the single signed proxy $\pi$. The standard-basket approximation in Proposition~\ref{cor2.8} is assembled from several lognormal-sum probabilities, each fitted separately and with parameters that may change with the strike. Each fitted probability is valid and lies in $[0,1]$, but global monotonicity and convexity of the assembled standard-basket price are not automatic. The numerical study should therefore check strike monotonicity, discrete convexity, nonnegativity, and put--call parity over the full strike grid.
\end{remark}

\subsection{A tail-based interpretation of the pricing error}
\label{subsec:tail-error}

Moment matching controls a finite number of global distributional features. It does not, by itself, guarantee an accurate option price. Two random variables can have the same first four moments and still assign different probabilities to the region above the strike. Since a call option pays only in that region, the relevant question is not only whether the moments are close, but also whether the approximating distribution reproduces the upper part of the basket distribution.

We first recall a standard stop-loss identity. It is stated explicitly because it provides the link between a distributional approximation and the corresponding option-price error.

\begin{lemma}[Stop-loss representation]
\label{lem:stop-loss-representation}
Let $X$ be an integrable random variable with distribution function $F_X(x)=P(X\leq x)$. Then, for every $K\in\mathbb R$,
\begin{equation}
\label{eq:stop-loss-representation}
E[(X-K)^+]
=
\int_K^\infty P(X>x)\,dx
=
\int_K^\infty \bigl[1-F_X(x)\bigr]\,dx.
\end{equation}
\end{lemma}

\begin{proof}
For every real number $x$,
\[
(x-K)^+
=
\int_K^\infty \mathbf{1}_{\{x>u\}}\,du.
\]
Substituting $x=X$ gives a nonnegative random integral. Tonelli's theorem therefore allows us to exchange the expectation and the integral:
\[
E[(X-K)^+]
=
\int_K^\infty E[\mathbf{1}_{\{X>u\}}]\,du
=
\int_K^\infty P(X>u)\,du.
\]
Since $F_X(u)=P(X\leq u)$, we have $P(X>u)=1-F_X(u)$, which proves \eqref{eq:stop-loss-representation}.
\end{proof}

Let
\[
C_B(K)=e^{-rT}E[(B(T)-K)^+]
\]
denote the true general-basket call price and let
\[
\widehat C_\pi(K)=e^{-rT}E[(\pi-K)^+]
\]
denote the price obtained from an admissible signed proxy $\pi$.

\begin{proposition}[Integrated CDF representation of the pricing error]
\label{prop:integrated-cdf-error}
Suppose that $B(T)$ and $\pi$ are integrable. Then
\begin{equation}
\label{eq:integrated-cdf-error}
C_B(K)-\widehat C_\pi(K)
=
e^{-rT}
\int_K^\infty
\left[
F_\pi(x)-F_B(x)
\right]dx,
\end{equation}
where $F_B$ and $F_\pi$ denote the distribution functions of $B(T)$ and $\pi$, respectively. Consequently,
\begin{equation}
\label{eq:tail-error-bound}
\left|C_B(K)-\widehat C_\pi(K)\right|
\leq
e^{-rT}
\int_K^\infty
\left|F_\pi(x)-F_B(x)\right|dx.
\end{equation}
\end{proposition}

\begin{proof}
Applying Lemma~\ref{lem:stop-loss-representation} to $B(T)$ and $\pi$ gives
\[
C_B(K)
=
e^{-rT}\int_K^\infty [1-F_B(x)]\,dx
\]
and
\[
\widehat C_\pi(K)
=
e^{-rT}\int_K^\infty [1-F_\pi(x)]\,dx.
\]
Subtracting the second expression from the first yields
\[
C_B(K)-\widehat C_\pi(K)
=
e^{-rT}\int_K^\infty [F_\pi(x)-F_B(x)]\,dx,
\]
which is \eqref{eq:integrated-cdf-error}. The bound in \eqref{eq:tail-error-bound} follows from the triangle inequality.
\end{proof}

The proposition motivates the strike-dependent integrated tail discrepancy
\begin{equation}
\label{eq:integrated-tail-discrepancy}
\mathcal E_{\mathrm{tail}}(K)
=
\int_K^\infty
\left|F_\pi(x)-F_B(x)\right|dx.
\end{equation}
It follows that
\[
\left|C_B(K)-\widehat C_\pi(K)\right|
\leq
e^{-rT}\mathcal E_{\mathrm{tail}}(K).
\]
Unlike a density comparison over the whole support, $\mathcal E_{\mathrm{tail}}(K)$ focuses only on the part of the distribution that affects the call payoff at strike $K$.

\begin{corollary}[Direction of the pricing bias]
\label{cor:direction-pricing-bias}
Under the assumptions of Proposition~\ref{prop:integrated-cdf-error}:
\begin{enumerate}
\item if $F_\pi(x)\geq F_B(x)$ for all $x\geq K$, then $\widehat C_\pi(K)\leq C_B(K)$;
\item if $F_\pi(x)\leq F_B(x)$ for all $x\geq K$, then $\widehat C_\pi(K)\geq C_B(K)$.
\end{enumerate}
\end{corollary}

\begin{proof}
The conclusions follow directly from the sign of the integrand in \eqref{eq:integrated-cdf-error}.
\end{proof}

The first four moments of $B(T)$ and $\pi$ are equal, but this does not force the integral in \eqref{eq:integrated-cdf-error} to be small. The following observation makes this point precise.

\begin{proposition}[Moment matching and signed CDF area]
\label{prop:matched-mean-cdf-area}
Suppose that $B(T)$ and $\pi$ are integrable and have the same mean. Then
\begin{equation}
\label{eq:matched-mean-cdf-area}
\int_{-\infty}^{\infty}
\left[
F_\pi(x)-F_B(x)
\right]dx
=
0.
\end{equation}
Therefore, positive and negative CDF discrepancies can cancel over the whole real line even when the discrepancy above a particular strike is economically important.
\end{proposition}

\begin{proof}
For any integrable random variable $X$,
\[
E[X]
=
\int_0^\infty [1-F_X(x)]\,dx
-
\int_{-\infty}^{0}F_X(x)\,dx.
\]
Applying this identity to $B(T)$ and $\pi$ and subtracting gives
\[
E[B(T)]-E[\pi]
=
\int_{-\infty}^{\infty}
[F_\pi(x)-F_B(x)]\,dx.
\]
The left-hand side is zero because the first moment is matched.
\end{proof}

Proposition~\ref{prop:matched-mean-cdf-area} explains why matching the mean, variance, skewness, and kurtosis cannot replace a tail diagnostic. The moment conditions constrain the distribution globally, while an option price depends on a one-sided region determined by the strike.

\begin{proposition}[Local relation between CDF error and price error]
\label{prop:local-price-error}
Define
\[
\Delta_C(K)=C_B(K)-\widehat C_\pi(K).
\]
At every strike $K$ at which both $F_B$ and $F_\pi$ are continuous,
\begin{equation}
\label{eq:price-error-slope}
\Delta_C'(K)
=
e^{-rT}\bigl[F_B(K)-F_\pi(K)\bigr].
\end{equation}
If both distributions admit densities $f_B$ and $f_\pi$ that are continuous at $K$, then
\begin{equation}
\label{eq:price-error-curvature}
\Delta_C''(K)
=
e^{-rT}\bigl[f_B(K)-f_\pi(K)\bigr].
\end{equation}
\end{proposition}

\begin{proof}
Equation~\eqref{eq:integrated-cdf-error} writes $\Delta_C(K)$ as an integral with lower limit $K$. At a continuity point of the integrand, the fundamental theorem of calculus gives
\[
\Delta_C'(K)
=
-e^{-rT}[F_\pi(K)-F_B(K)],
\]
which is \eqref{eq:price-error-slope}. Differentiating once more at a point where both densities are continuous gives \eqref{eq:price-error-curvature}.
\end{proof}

Equation~\eqref{eq:price-error-slope} shows that the CDF error at the strike determines the slope of the pricing error across strikes. Equation~\eqref{eq:price-error-curvature} shows that the density error determines the curvature of the pricing error. Thus, density plots are informative, but the CDF and integrated tail discrepancy are more directly connected to the option value.

\begin{corollary}[Global Wasserstein bound]
\label{cor:wasserstein-bound}
If $B(T)$ and $\pi$ have finite first moments, then
\begin{equation}
\label{eq:wasserstein-bound}
\left|C_B(K)-\widehat C_\pi(K)\right|
\leq
e^{-rT}
\int_{-\infty}^{\infty}
\left|F_\pi(x)-F_B(x)\right|dx.
\end{equation}
In one dimension, the integral on the right-hand side is the first Wasserstein distance between the two distributions.
\end{corollary}

\begin{proof}
The result follows because the integral in \eqref{eq:tail-error-bound} is taken over the half-line $\{x\in\R:x\geq K\}$, which is a subset of the real line.
\end{proof}

\subsection{CDF-error decomposition for standard baskets}
\label{subsec:standard-basket-cdf-error}

The previous results apply directly to the signed proxy for a general basket. The standard-basket construction has a different structure: its price is a linear combination of several probabilities, and each probability is approximated separately.

For $i=1,\ldots,n+1$, define
\[
L_i
=
\sum_{j=1}^{n}
e^{N_j(m_{ij},\bar\sigma_j^2)}
\]
and let
\[
F_i(z)=P(L_i\leq z).
\]
Let $\widehat F_i$ denote the CDF obtained from the corresponding four-moment shifted-mixture approximation. Recall the notation $q_1=1$, $q_i=-1$ for $i=2,\ldots,n+1$, $s_i=|\omega_i|S_i(0)$ for $i\leq n$, and $s_{n+1}=Ke^{-rT}$.

\begin{proposition}[Exact decomposition of the standard-basket pricing error]
\label{prop:standard-basket-cdf-error}
The exact and approximate standard-basket prices can be written as
\[
C^{sd}(T)
=
\sum_{i=1}^{n+1}q_i s_iF_i(1)
\]
and
\[
\widehat C^{sd}(T)
=
\sum_{i=1}^{n+1}q_i s_i\widehat F_i(1).
\]
Therefore,
\begin{equation}
\label{eq:standard-basket-cdf-error}
C^{sd}(T)-\widehat C^{sd}(T)
=
\sum_{i=1}^{n+1}
q_i s_i
\left[
F_i(1)-\widehat F_i(1)
\right],
\end{equation}
and
\begin{equation}
\label{eq:standard-basket-cdf-bound}
\left|C^{sd}(T)-\widehat C^{sd}(T)\right|
\leq
\sum_{i=1}^{n+1}
s_i
\left|
F_i(1)-\widehat F_i(1)
\right|.
\end{equation}
\end{proposition}

\begin{proof}
Proposition~\ref{cor2.8} gives the exact representation
\[
C^{sd}(T)=\sum_{i=1}^{n+1}q_i s_iF_i(1).
\]
Replacing every $F_i$ with its four-moment approximation $\widehat F_i$ gives the approximate price. Subtracting the two formulas proves \eqref{eq:standard-basket-cdf-error}. Applying the triangle inequality and using $|q_i|=1$ gives \eqref{eq:standard-basket-cdf-bound}.
\end{proof}

The decomposition in \eqref{eq:standard-basket-cdf-error} identifies the contribution of each probability approximation to the final price. Define
\begin{equation}
\label{eq:standard-error-contribution}
\Delta_i
=
q_i s_i
\left[
F_i(1)-\widehat F_i(1)
\right].
\end{equation}
Then
\[
C^{sd}(T)-\widehat C^{sd}(T)
=
\sum_{i=1}^{n+1}\Delta_i.
\]
The individual errors $\Delta_i$ may have opposite signs and partly cancel. A small total price error therefore does not necessarily imply that every CDF approximation is accurate. For this reason, the numerical analysis reports both the total pricing error and the componentwise CDF errors at the threshold $1$ whenever the required simulation output is available.

\subsection{A tail-sensitive parameterization of the mixture}
\label{subsec:tail-sensitive-mixture}

The positive-sum approximation uses two variance states, $\alpha$ and $\beta$. Treating them as unspecified constants weakens reproducibility. A convenient parameterization is
\begin{equation}
\label{eq:symmetric-mixture-parameterization}
\alpha_\delta=1-\delta,
\qquad
\beta_\delta=1+\delta,
\qquad
0<\delta<1.
\end{equation}
The restriction $0<\delta<1$ keeps both variance states positive. For each fixed $\delta$, Proposition~\ref{32} determines $p$, $s$, $m$, and $\tau$, subject to the admissibility conditions in Definition~\ref{def:admissible-positive-sum}. Thus, $\delta$ controls the separation between the two conditional lognormal components without changing the analytical form of the CDF.

Let $\mathcal D_{\mathrm{adm}}\subset(0,1)$ denote the set of values of $\delta$ for which an admissible four-moment solution exists. Since the first four moments are already matched, a natural model-based criterion uses the standardized fifth central moment
\begin{equation}
\label{eq:fifth-standardized-moment}
\gamma_5(X)
=
\frac{E[(X-E[X])^5]}{\Var(X)^{5/2}}.
\end{equation}
One may select
\begin{equation}
\label{eq:delta-fifth-moment}
\delta^\star
\in
\arg\min_{\delta\in\mathcal D_{\mathrm{adm}}}
\left|
\gamma_5(\bar S)-\gamma_5(W_\delta)
\right|.
\end{equation}

This criterion uses information beyond kurtosis while preserving the analytical pricing step. It also requires only a one-dimensional search over $\delta$. A tail-based alternative is
\begin{equation}
\label{eq:delta-tail-criterion}
\delta^\star_{\mathrm{tail}}
\in
\arg\min_{\delta\in\mathcal D_{\mathrm{adm}}}
\sum_{\ell=1}^{L}\omega_\ell
\int_{k_\ell}^{\infty}
\left|
F_{W_\delta}(x)-F_{\bar S}(x)
\right|dx,
\end{equation}
where $\omega_\ell\geq0$ are strike weights. Because $F_{\bar S}$ is generally unavailable in closed form, \eqref{eq:delta-tail-criterion} requires a numerical benchmark and is best used as a robustness diagnostic rather than as the default analytical calibration rule.

\begin{remark}[Baseline and refined specifications]
\label{rem:baseline-refined-mixture}
The fixed-$(\alpha,\beta)$ and the optimized-$\delta$ constructions should not be mixed in the same numerical comparison. Results based on \eqref{eq:delta-fifth-moment} or \eqref{eq:delta-tail-criterion} require the moment equations and all option prices to be recomputed. The current pricing tables should therefore be interpreted as results for the baseline fixed-$(\alpha,\beta)$ specification unless the refined values of $\delta$ are reported explicitly.
\end{remark}

The theoretical results in this section guide the numerical study. Admissibility determines whether the fitted proxy is valid, the financial-properties results provide strike-based consistency checks, and Propositions~\ref{prop:integrated-cdf-error} and~\ref{prop:standard-basket-cdf-error} identify the distributional quantities that directly control the pricing error.

\section{Numerical Evaluation}
\label{sec:numerical-evaluation}

This section has two purposes. First, it preserves the six standard-basket
benchmarks used to evaluate the exact probability reformulation. Second, it
studies the signed general-basket proxy in an empirical crack-spread
application. Monte Carlo simulation is the numerical benchmark in both parts.
For an approximation $\widehat C_j$ and Monte Carlo value $C_j^{MC}$, we report
\[
\mathrm{MAE}=\frac1N\sum_{j=1}^{N}|\widehat C_j-C_j^{MC}|,
\qquad
\mathrm{RMSE}=
\left[\frac1N\sum_{j=1}^{N}(\widehat C_j-C_j^{MC})^2\right]^{1/2}.
\]
Monte Carlo uncertainty is measured by the standard error of the discounted
sample mean. Whenever CDF or integrated-tail diagnostics are reported, they use
the same simulated terminal basket values as the corresponding price. This
keeps the pricing and distributional comparisons internally consistent.

\subsection{Benchmark three-asset standard baskets}
\label{subsec:threeasset-standard-examples}

We retain the six three-asset examples, their input parameters, and their
weight vector exactly as in the original numerical study. The baseline uses
$T=1$, $r=0.03$, and the fixed variance states $\alpha=0.9$ and $\beta=1.1$.
These examples separate the gain from the exact probability reformulation from
the additional gain obtained by matching the fourth moment.

Each contract uses the weight vector
\[
\left(\frac{2}{3},-\frac{1}{3},-1\right),
\]
so it contains one positive component and two negative components, as required
by the standard-basket definition in
Section~\ref{sec:standard-basket-framework}. The time-zero call price is
\begin{equation}
\label{threeasset-payoff}
C_T^{\mathrm{3SB}}
=
e^{-rT}\E\!\left[
\left(
\frac{2}{3}S_T^1-\frac{1}{3}S_T^2-S_T^3-K
\right)^+
\right].
\end{equation}

Applying Proposition~\ref{cor2.8}, we obtain
\begin{equation}\label{threeasset-representation}
\begin{split}
C_T^{\mathrm{3SB}}&=s_1P(\sum_{j=1}^3e^{N_j(m_{1j},\bar{\sigma}_j^2)}\le1)-s_2P(\sum_{j=1}^3e^{N_j(m_{2j},\bar{\sigma}_j^2)}\le1)\\
&-s_3P(\sum_{j=1}^3e^{N_j(m_{3j},\bar{\sigma}_j^2)}\le1)-Ke^{-rT}P(\sum_{j=1}^3e^{N_j(m_{4j},\bar{\sigma}_j^2)}\le1),
\end{split}
\end{equation}
where $s_1=\frac{2}{3}S_0^1$, $s_2=\frac{1}{3}S_0^2$, $s_3=S_0^3$, and
\begin{equation}
\begin{split}
\bar{\sigma}_1=&\sigma_1\sqrt{T},\quad \bar{\sigma}_2=\sqrt{\sigma_2^2+\sigma_1^2-2\sigma_1\sigma_2\rho_{12}}\sqrt{T}, \quad \bar{\sigma}_3=\sqrt{\sigma_3^2+\sigma_1^2-2\sigma_1\sigma_3\rho_{13}}\sqrt{T},\\
m_{11}=&\ln (Ke^{-rT})-\ln s_1-\frac{1}{2}\sigma_1^2T,\\
m_{21}=&\ln (Ke^{-rT})-\ln s_1+\frac{1}{2}\sigma_1^2T-\rho_{12}\sigma_2\sigma_1T,\\
m_{31}=&\ln (Ke^{-rT})-\ln s_1+\frac{1}{2}\sigma_1^2T-\rho_{13}\sigma_3\sigma_1T,\\
m_{41}=&\ln(Ke^{-rT})-\ln s_1+\frac{1}{2}\sigma_1^2T,\\
m_{12}=&\ln s_2-\ln s_1-\frac{1}{2}\sigma_1^2T-\frac{1}{2}\sigma_2^2T+\rho_{12}\sigma_1\sigma_2T,\\
m_{22}=&\ln s_2-\ln s_1+\frac{1}{2}\sigma_1^2T+\frac{1}{2}\sigma_2^2T-\rho_{12}\sigma_1\sigma_2T,\\
m_{32}=&\ln s_2-\ln s_1+\frac{1}{2}\sigma_1^2T-\frac{1}{2}\sigma_2^2T+\rho_{23}\sigma_2\sigma_3T-\rho_{13}\sigma_1\sigma_3T,\\
m_{42}=&\ln s_2-\ln s_1+\frac{1}{2}\sigma_1^2T-\frac{1}{2}\sigma_2^2T,\\
m_{13}=&\ln s_3-\ln s_1-\frac{1}{2}\sigma_1^2T-\frac{1}{2}\sigma_3^2T+\rho_{13}\sigma_1\sigma_3T,\\
m_{23}=&\ln s_3-\ln s_1+\frac{1}{2}\sigma_1^2T-\frac{1}{2}\sigma_3^2T+\rho_{23}\sigma_2\sigma_3T-\rho_{12}\sigma_1\sigma_2T,\\
m_{33}=&\ln s_3-\ln s_1+\frac{1}{2}\sigma_1^2T+\frac{1}{2}\sigma_3^2T-\rho_{13}\sigma_1\sigma_3T,\\
m_{43}=&\ln s_3-\ln s_1+\frac{1}{2}\sigma_1^2T-\frac{1}{2}\sigma_3^2T,
\end{split}
\end{equation}
and the correlation matrix $\bar{\Gamma}=(\theta_{jk})$ of the normal random variables $(N_1, N_2, N_3)$ is given by
\begin{equation}
\label{eq:threeasset-correlation}
\begin{aligned}
\theta_{11}&=\theta_{22}=\theta_{33}=1,\\
\theta_{12}=\theta_{21}
&=\frac{\sigma_1-\sigma_2\rho_{12}}
{\sqrt{\sigma_1^2+\sigma_2^2-2\sigma_1\sigma_2\rho_{12}}},\\
\theta_{13}=\theta_{31}
&=\frac{\sigma_1-\sigma_3\rho_{13}}
{\sqrt{\sigma_1^2+\sigma_3^2-2\sigma_1\sigma_3\rho_{13}}},\\
\theta_{23}=\theta_{32}
&=\frac{\sigma_2\sigma_3\rho_{23}
-\sigma_1\sigma_2\rho_{12}
-\sigma_1\sigma_3\rho_{13}+\sigma_1^2}
{\sqrt{\sigma_2^2+\sigma_1^2-2\sigma_1\sigma_2\rho_{12}}
 \sqrt{\sigma_3^2+\sigma_1^2-2\sigma_1\sigma_3\rho_{13}}}.
\end{aligned}
\end{equation}

For $i=1,2,3,4$, define the normalized lognormal sum
\[
\bar S_i
=
\frac{1}{a_i}
\sum_{j=1}^{3}
\exp\!\left(N_j(m_{ij},\bar\sigma_j^2)\right),
\qquad
a_i=\sum_{j=1}^{3}e^{m_{ij}},
\]
and let
\[
W^{(i)}=e^{s^{(i)}\sqrt{Y^{(i)}}N+m^{(i)}}+\tau^{(i)}
\]
denote its four-moment approximation. Since
$\sum_{j=1}^{3}\exp(N_j(m_{ij},\bar\sigma_j^2))=a_i\bar S_i$,
the threshold $1$ in \eqref{threeasset-representation} becomes
\[
z_i=a_i^{-1},
\qquad i=1,2,3,4.
\]
The distribution of $Y^{(i)}$ is
$\PP(Y^{(i)}=\alpha)=p^{(i)}$ and
$\PP(Y^{(i)}=\beta)=1-p^{(i)}$. Define
\begin{equation}
\label{eq:standard-proxy-cdf-term}
\widehat F_i(1)=
\1_{\{z_i>\tau^{(i)}\}}
\left[
 p^{(i)}\Phi\!\left(
 \frac{\ln(z_i-\tau^{(i)})-m^{(i)}}{s^{(i)}\sqrt{\alpha}}
 \right)
 +(1-p^{(i)})\Phi\!\left(
 \frac{\ln(z_i-\tau^{(i)})-m^{(i)}}{s^{(i)}\sqrt{\beta}}
 \right)
\right].
\end{equation}

\begin{theorem}[Three-asset standard-basket approximation]
\label{thm:threeasset-approximation}
The option price in \eqref{threeasset-representation} is approximated by
\begin{equation}
\label{approx_price}
\widehat C_T^{\mathrm{3SB}}
=
s_1\widehat F_1(1)
-s_2\widehat F_2(1)
-s_3\widehat F_3(1)
-Ke^{-rT}\widehat F_4(1),
\end{equation}
where each $\widehat F_i(1)$ is given by
\eqref{eq:standard-proxy-cdf-term}.
\end{theorem}

\begin{proof}
For each $i$, the exact probability is
$\PP(a_i\bar S_i\leq1)=\PP(\bar S_i\leq z_i)$. Replace $\bar S_i$ by
$W^{(i)}$ and apply Lemma~\ref{lemma31}. This gives
\eqref{eq:standard-proxy-cdf-term}. Substitution into
\eqref{threeasset-representation} yields \eqref{approx_price}.
\end{proof}

\begin{table}[htbp]
\centering
\begin{threeparttable}
\caption{Input parameters for the three-asset standard baskets}
\label{tab:threeasset-inputs}
\scriptsize
\setlength{\tabcolsep}{1.5pt}
\renewcommand{\arraystretch}{1.14}
\begin{tabular}{>{\raggedright\arraybackslash}p{2.15cm}*{6}{>{\centering\arraybackslash}p{1.95cm}}}
\toprule
& \textbf{Basket 1} & \textbf{Basket 2} & \textbf{Basket 3} &
\textbf{Basket 4} & \textbf{Basket 5} & \textbf{Basket 6}\\
\midrule
Initial values $S_i(0)$
& $[150,90,60]$ & $[120,60,50]$ & $[210,120,70]$
& $[180,80,60]$ & $[200,80,60]$ & $[140,70,55]$\\
\rowrule
Volatilities $\sigma_i$
& {\tiny $[0.20,0.30,0.25]$} & {\tiny $[0.10,0.20,0.30]$} & {\tiny $[0.10,0.10,0.20]$}
& {\tiny $[0.10,0.15,0.30]$} & {\tiny $[0.15,0.20,0.20]$} & {\tiny $[0.10,0.15,0.20]$}\\
\rowrule
Correlations
& \makecell{$\rho_{12}=0.9$\\$\rho_{13}=0.8$\\$\rho_{23}=0.9$}
& \makecell{$\rho_{12}=0.5$\\$\rho_{13}=0.6$\\$\rho_{23}=0.7$}
& \makecell{$\rho_{12}=0.1$\\$\rho_{13}=0.2$\\$\rho_{23}=0.3$}
& \makecell{$\rho_{12}=0.7$\\$\rho_{13}=0.8$\\$\rho_{23}=0.9$}
& \makecell{$\rho_{12}=0.5$\\$\rho_{13}=0.7$\\$\rho_{23}=0.9$}
& \makecell{$\rho_{12}=0.6$\\$\rho_{13}=0.7$\\$\rho_{23}=0.9$}\\
\rowrule
Strike $K$ & $10$ & $15$ & $20$ & $18$ & $40$ & $12$\\
\bottomrule
\end{tabular}
\begin{tablenotes}\footnotesize
\item All examples use $T=1$, $r=0.03$, and the weight vector
$(2/3,-1/3,-1)$.
\end{tablenotes}
\end{threeparttable}
\end{table}

\begin{table}[htbp]
\centering
\begin{threeparttable}
\caption{Three-asset standard-basket option prices}
\label{tab:threeasset-prices}
\scriptsize
\setlength{\tabcolsep}{4.3pt}
\renewcommand{\arraystretch}{1.14}
\begin{tabular}{>{\raggedright\arraybackslash}p{5.0cm}*{6}{>{\centering\arraybackslash}p{1.46cm}}}
\toprule
\textbf{Method} & \textbf{Basket 1} & \textbf{Basket 2} & \textbf{Basket 3}
& \textbf{Basket 4} & \textbf{Basket 5} & \textbf{Basket 6}\\
\midrule
Direct three-moment match
& 4.4766 & 3.0452 & 12.4475 & 15.4661 & 8.8169 & 5.3425\\
\rowrule
Three-moment match after the probability reformulation
& 4.3709 & 3.0210 & 12.3925 & 15.4427 & 8.7566 & 5.3132\\
\rowrule
Four-moment probability method, Theorem~\ref{thm:threeasset-approximation}
& 4.3740 & 3.0237 & 12.4093 & 15.4491 & 8.7628 & 5.3158\\
\rowrule
Monte Carlo benchmark
& \makecell{4.3738\\(0.0020)}
& \makecell{3.0263\\(0.0017)}
& \makecell{12.4077\\(0.0042)}
& \makecell{15.4472\\(0.0035)}
& \makecell{8.7620\\(0.0035)}
& \makecell{5.3155\\(0.0020)}\\
\bottomrule
\end{tabular}
\begin{tablenotes}\footnotesize
\item Monte Carlo standard errors are shown in parentheses. The direct and
probability-based three-moment rows use the shifted lognormal method of
\cite{hu2023pricingbasketoptionsmoments}.
\end{tablenotes}
\end{threeparttable}
\end{table}

Relative to Monte Carlo, the mean absolute error is $0.0437$ for the direct
three-moment approximation, $0.0059$ for the three-moment approximation applied
after the probability reformulation, and $0.0012$ for the proposed four-moment
method. The corresponding maximum absolute errors are $0.1028$, $0.0152$, and
$0.0026$. These results identify two distinct gains. The exact probability
reformulation produces the largest improvement, and matching the fourth moment
then reduces the remaining error further. The gain therefore does not arise
only from adding one moment; it also depends on approximating the positive
lognormal sums that appear in the exact pricing identity.

\subsubsection{Distributional assessment of the four probability terms}

\label{subsubsec:standard-distribution-diagnostics}

Table~\ref{tab:threeasset-prices} evaluates the accuracy of the final option
price. To understand why the probability-based method performs well, it is
also useful to examine the four distributions that enter the exact pricing
representation. For Basket~1, define
\[
L_i=\sum_{j=1}^{3}\exp\!\left(N_j(m_{ij},\bar\sigma_j^2)\right),
\qquad i=1,2,3,4,
\]
where the first three indices correspond to the three asset contributions and
the fourth corresponds to the discounted strike contribution. The standard
basket price depends on these variables through
\begin{equation}
\label{eq:basket1-four-cdf-terms}
C^{\mathrm{3SB}}(T)
=
\sum_{i=1}^{4}q_i s_i F_i(1),
\qquad
F_i(x)=P(L_i\leq x).
\end{equation}
Thus, the four panels in Figure~\ref{fig:standard-basket1-density-comparison}
are not four different basket contracts. They are the four probability terms
needed to assemble the price of the same contract.

\begin{figure}[htbp]
    \centering
    \subfloat[Positive-asset probability]{%
        \includegraphics[width=0.48\textwidth]{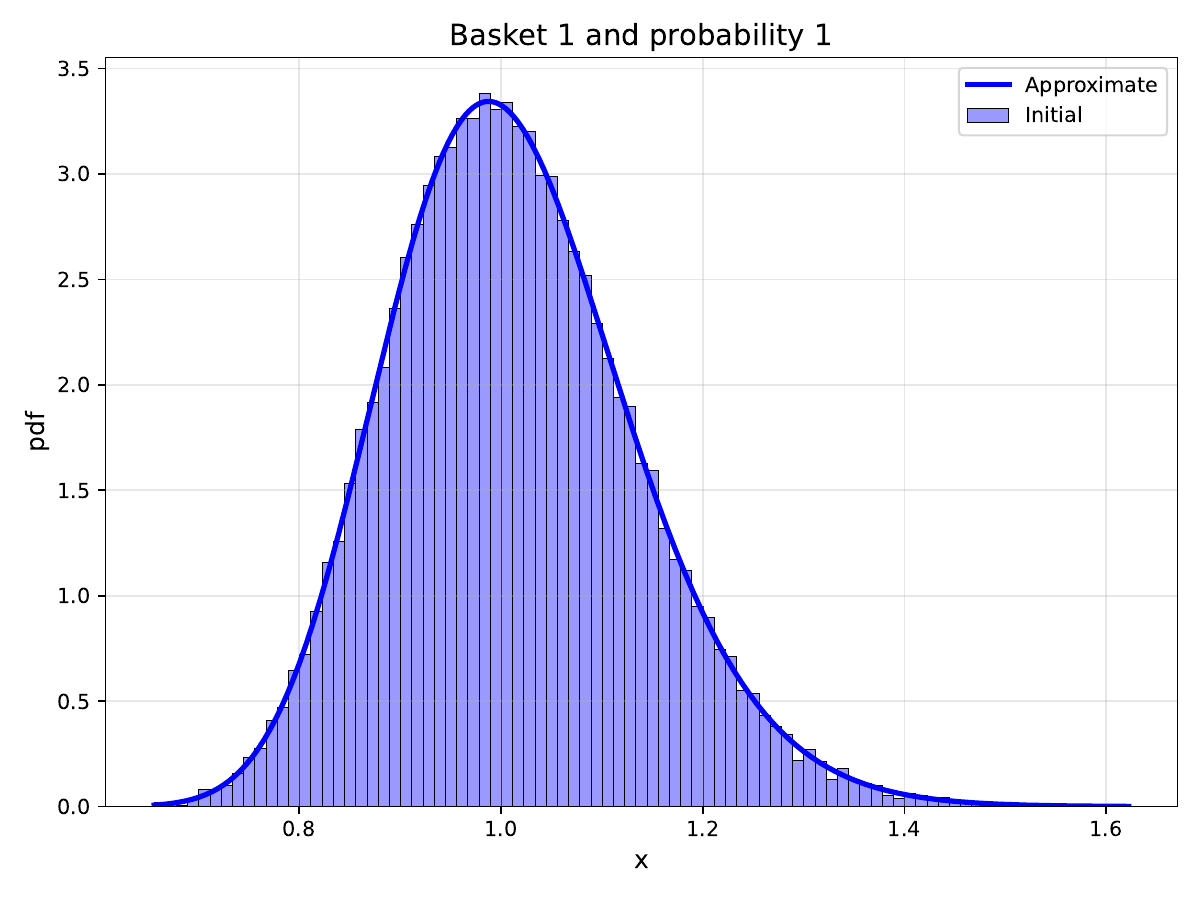}}
    \hfill
    \subfloat[First negative-asset probability]{%
        \includegraphics[width=0.48\textwidth]{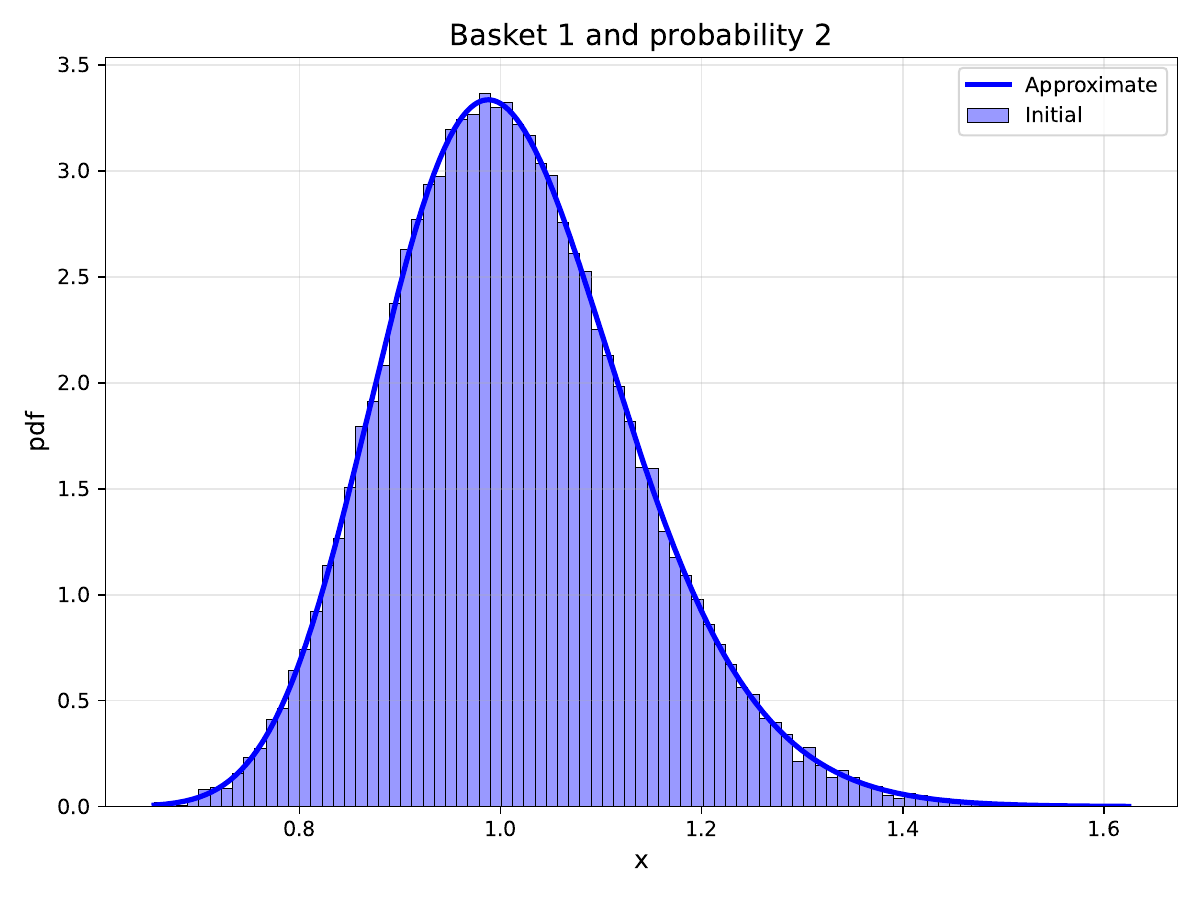}}

    \vspace{0.35cm}

    \subfloat[Second negative-asset probability]{%
        \includegraphics[width=0.48\textwidth]{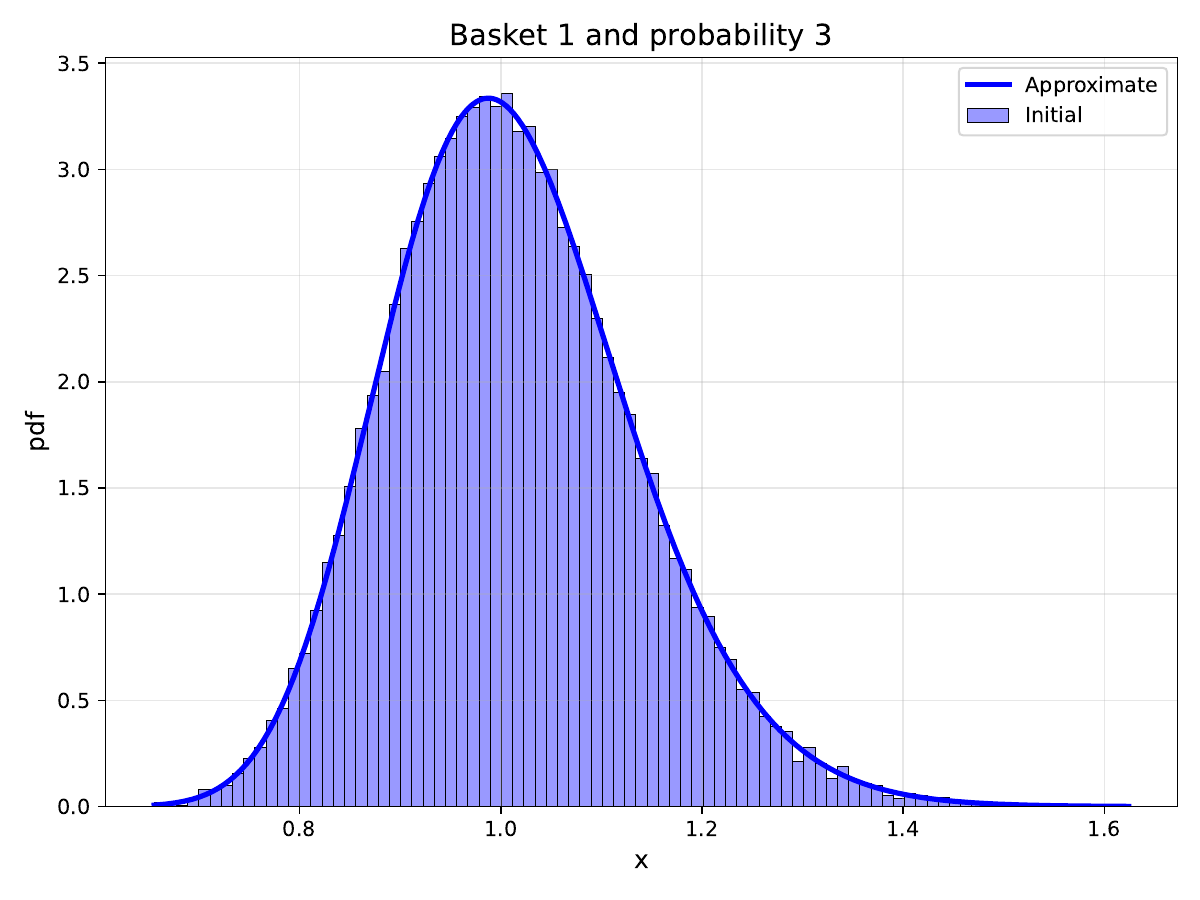}}
    \hfill
    \subfloat[Discounted-strike probability]{%
        \includegraphics[width=0.48\textwidth]{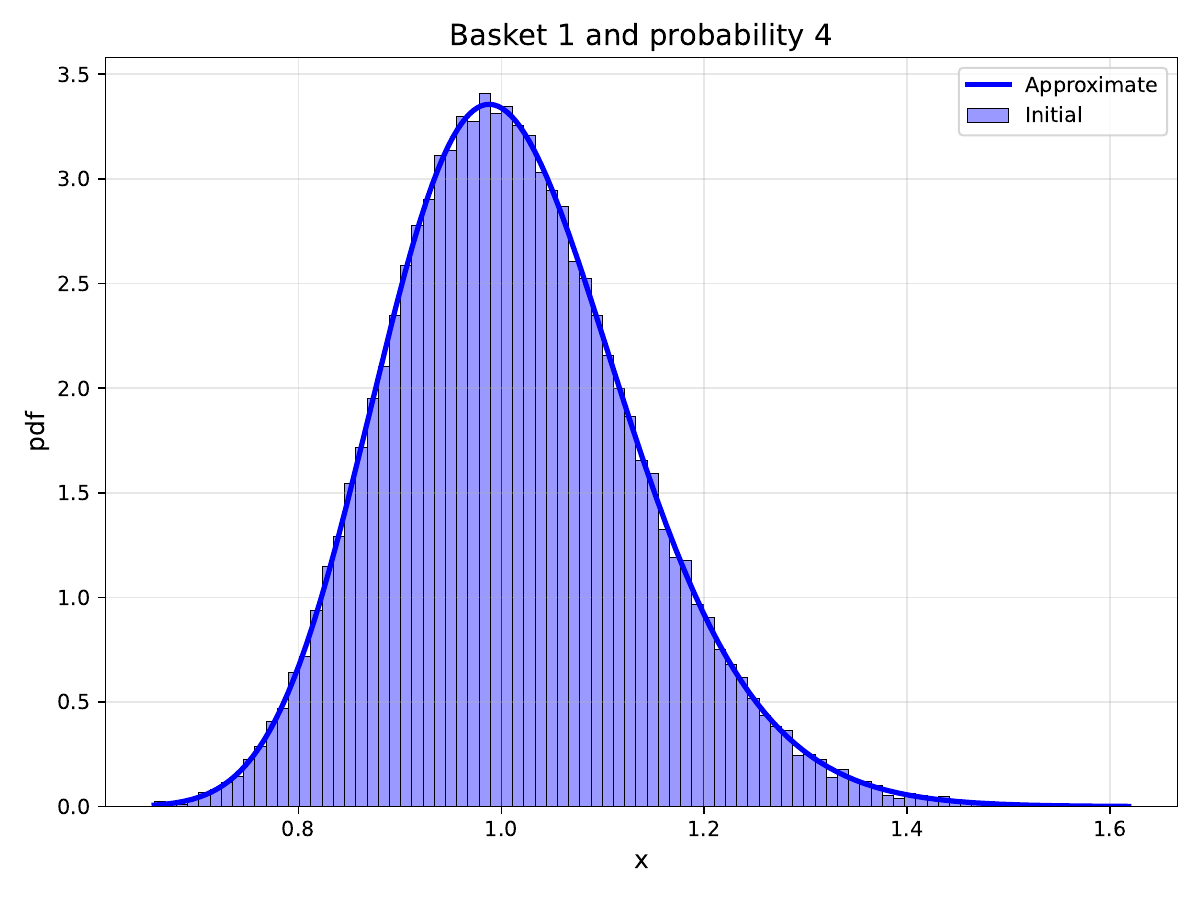}}

    \caption{Distributional assessment of the four correlated lognormal sums
    entering the price of Standard Basket~1. In each panel, the curve labeled
    ``approximate'' represents the benchmark density, and the shaded initial
    represents the fitted four-moment proxy. The option price uses the four CDF
    values at the common threshold $x=1$, as shown in
    \eqref{eq:basket1-four-cdf-terms}.}
    \label{fig:standard-basket1-density-comparison}
\end{figure}

The fitted distributions reproduce the location, dispersion, and right-skewed
shape of all four benchmark distributions closely. The visual agreement is
strong around $x=1$, the threshold that enters the option price. This supports
the probability-level construction and helps explain the small four-moment
error in Table~\ref{tab:threeasset-prices}. The four panels also indicate that
the final result is not driven by an accurate fit to only one term: each of the
four probability components is represented well over the region surrounding
the pricing threshold.

The figure should nevertheless be interpreted as a diagnostic, not as a
stand-alone proof of price accuracy. A density comparison describes the fit
over the displayed support, whereas Proposition~\ref{prop:standard-basket-cdf-error}
shows that the price error is determined by the signed and weighted CDF errors
at exactly $x=1$:
\[
C^{sd}(T)-\widehat C^{sd}(T)
=
\sum_{i=1}^{4}q_i s_i\bigl[F_i(1)-\widehat F_i(1)\bigr].
\]
Small errors in different terms may partly offset each other. For this reason,
the price table and the CDF-error decomposition remain the primary valuation
checks, while Figure~\ref{fig:standard-basket1-density-comparison} provides a
transparent distributional validation of the fitted proxies.

The role of this figure differs from the distributional diagnostics used later
for the empirical crack spread. A standard basket is priced from four separate
positive lognormal sums, so four fitted distributions must be examined. The
crack spread is treated directly as one mixed-sign basket and is represented by
one signed proxy; its diagnostics therefore compare the three- and four-moment
densities, CDFs, and upper-tail discrepancies for a single basket distribution.

\FloatBarrier

\subsection{Empirical general-basket application: the \texorpdfstring{$3{:}2{:}1$}{3:2:1} crack spread}
\label{subsec:crack-application}

The second experiment applies the signed proxy to an economically meaningful
mixed-sign basket. A long $3{:}2{:}1$ crack spread consists of two long RBOB gasoline
positions, one long ULSD position, and three short WTI crude-oil positions. RBOB and
ULSD futures are quoted in dollars per gallon, while WTI is quoted in dollars
per barrel. Since one barrel contains 42 gallons, the normalized spread value
in dollars per barrel is
\begin{equation}
\label{eq:crack-spread}
B_t^{\mathrm{cr}}
=
\frac{2}{3}\left(42F_t^{R}\right)
+
\frac{1}{3}\left(42F_t^{H}\right)
-
F_t^{C},
\end{equation}
where $F^R$, $F^H$, and $F^C$ denote RBOB, ULSD, and WTI futures,
respectively. This is the per-barrel version of the usual $3{:}2{:}1$
contract ratio \cite{cme2024crack}. Its weight vector is
\[
\left(\frac23,\frac13,-1\right),
\]
so it is a general basket rather than a standard basket under the definition
used in Section~\ref{sec:standard-basket-framework}.

\subsubsection{Data and parameter estimation}

Daily continuous futures series for the Yahoo Finance symbols \texttt{RB=F},
\texttt{HO=F}, and \texttt{CL=F} are used from January 2015 through July 2026
\cite{yahoo2026energy}. The three series are aligned by trading date. RBOB and
ULSD prices are multiplied by 42 before the basket is formed. Annualized
volatilities are estimated from daily log returns using 252 trading days, and
the dependence matrix is the sample correlation matrix of the aligned
returns.

The continuous WTI series contains a nonpositive observation during the April
2020 market disruption. A lognormal model and a log return are not defined at
a nonpositive price. We therefore remove returns that touch that observation.
No replacement or winsorization is used. This treatment is transparent, but
it also points to an important limitation: continuous front-month histories
contain roll effects and do not correspond to one fixed contract tenor.
Accordingly, the exercise below is a model-based study of the approximation,
not a replication of an exchange-traded option with matched futures and option
expiries.

Table~\ref{tab:crack-calibration} reports the fitted inputs. The latest aligned
spread value is
\[
B_0^{\mathrm{cr}}=63.2358,
\]
which is used as the baseline strike. The estimated correlations are positive and economically meaningful, but
they remain below one. The basket therefore retains substantial residual
uncertainty after the long refined-product positions are offset against crude
oil.

\begin{table}[htbp]
\centering
\begin{threeparttable}
\caption{Empirical inputs for the normalized $3{:}2{:}1$ crack-spread basket}
\label{tab:crack-calibration}
\renewcommand{\arraystretch}{1.14}
\begin{tabular}{lccc}
\toprule
Asset & Current level (\$/bbl) & Basket weight & Annualized volatility\\
\midrule
RBOB gasoline & 135.3072 & $2/3$ & 0.4489\\
\rowrule
ULSD          & 173.1030 & $1/3$ & 0.4028\\
\rowrule
WTI crude oil &  84.6700 & $-1$  & 0.4656\\
\bottomrule
\end{tabular}
\begin{tablenotes}\footnotesize
\item Pairwise correlations are
$\rho_{R,H}=0.6533$, $\rho_{R,C}=0.6877$, and
$\rho_{H,C}=0.7311$. Refined-product levels are converted from dollars per
gallon to dollars per barrel before estimation of the basket value.
\end{tablenotes}
\end{threeparttable}
\end{table}

Historical volatility and correlation estimates are used only as inputs to the
lognormal pricing experiment. They are not option-implied risk-neutral
parameters. All analytical prices and Monte Carlo benchmarks below are
therefore generated under the same fitted model. The comparison isolates the
error caused by the moment approximation.

\subsubsection{Baseline price comparison}

The baseline uses $T=1$, $r=0.03$, and $K=B_0^{\mathrm{cr}}$. Monte Carlo uses
$10^6$ paths. The reported uncertainty is the standard error of the discounted
sample mean. Table~\ref{tab:crack-baseline} shows that the four-moment method is
substantially closer to Monte Carlo than the three-moment method at the
baseline strike.

\begin{table}[htbp]
\centering
\caption{Baseline one-year crack-spread option prices}
\label{tab:crack-baseline}
\renewcommand{\arraystretch}{1.14}
\begin{tabular}{lcccc}
\toprule
Method & Option price & Absolute error & Error/MC s.e. & Fitted $p$\\
\midrule
Three moments & 15.1945 & 0.3752 & 13.64 & 1.0000\\
\rowrule
Four moments  & 14.8775 & 0.0583 &  2.12 & 0.9961\\
\rowrule
Monte Carlo   & \makecell{14.8192\\(0.0275)} & -- & -- & --\\
\bottomrule
\end{tabular}
\end{table}

Matching the fourth moment lowers the baseline absolute error by
approximately
\[
1-\frac{0.0583}{0.3752}=84.5\%.
\]
The remaining difference is about 2.12 Monte Carlo standard errors. It is
therefore more accurate to say that the four-moment estimate is substantially
closer to the benchmark, rather than statistically indistinguishable from it.
The fitted value $p=0.9961$ is close to the one-sided boundary. This indicates
that the positive branch dominates the fitted distribution, while the small
negative-sign probability is still needed to reproduce the fourth moment.

\subsubsection{Distributional fit at the baseline maturity}
\label{subsubsec:crack-distributional-fit}

Table~\ref{tab:crack-baseline} summarizes the pricing error, but it does not
show where the approximation differs from the simulated basket distribution.
We therefore examine the distributional fit at the one-year baseline. All
four panels in Figure~\ref{fig:crack-distributional-diagnostics} use the same
$10^6$ terminal basket values as the Monte Carlo price in
Table~\ref{tab:crack-baseline}. The analytical curves are computed from the
three- and four-moment parameters fitted to the same target moments. Thus, the
figure introduces no additional simulation sample and is directly comparable
with the baseline pricing results.

\begin{figure}[htbp]
\centering
\subfloat[Three-moment density fit]{%
\includegraphics[width=0.48\textwidth]{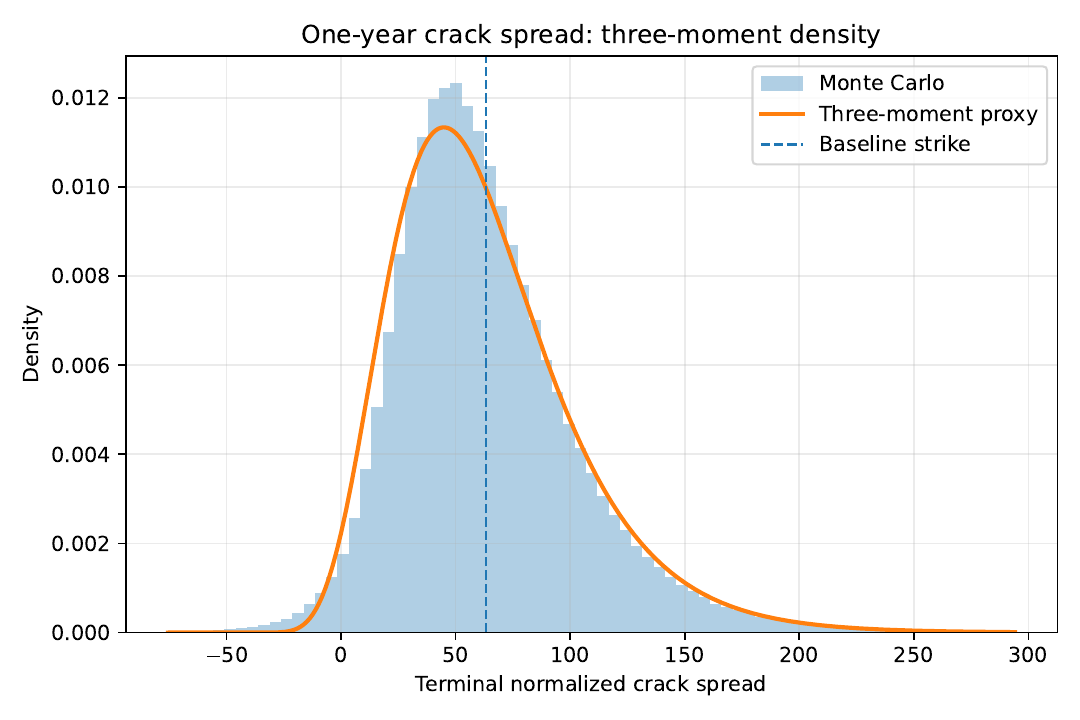}}
\hfill
\subfloat[Four-moment density fit]{%
\includegraphics[width=0.48\textwidth]{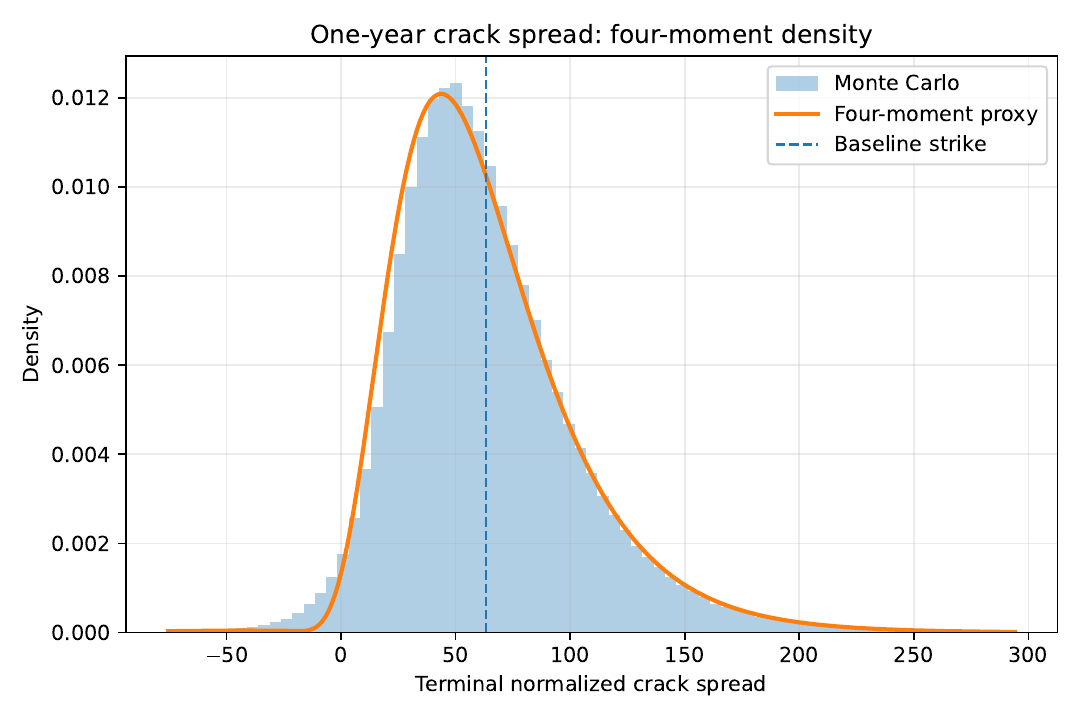}}

\medskip
\subfloat[Empirical and fitted CDFs]{%
\includegraphics[width=0.48\textwidth]{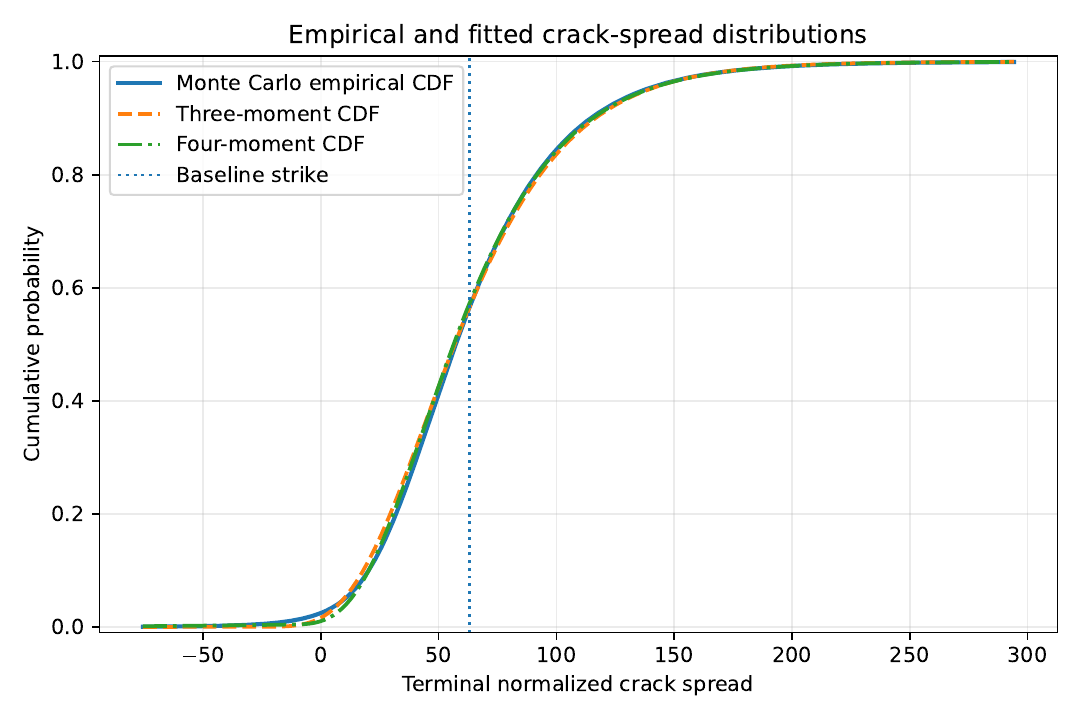}}
\hfill
\subfloat[Upper-tail CDF discrepancies]{%
\includegraphics[width=0.48\textwidth]{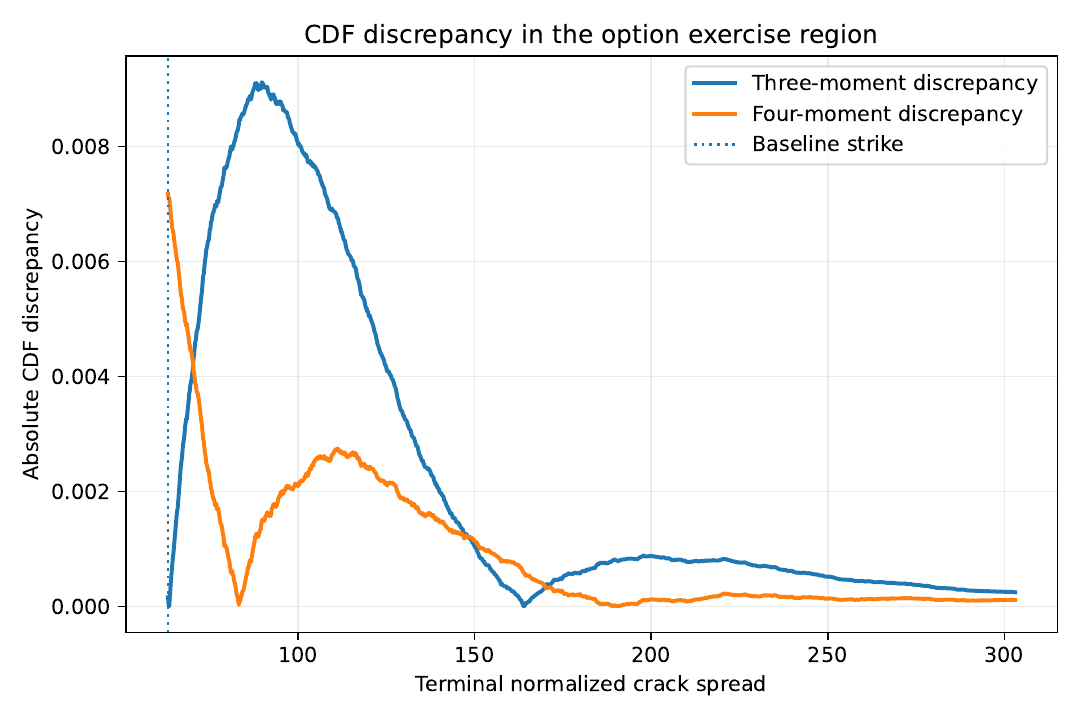}}
\caption{Distributional diagnostics for the one-year normalized
$3{:}2{:}1$ crack-spread basket. The vertical line marks the baseline strike
$K=B_0^{\mathrm{cr}}$. Panels (a) and (b) compare the Monte Carlo histogram
with the analytical densities implied by the three- and four-moment signed
proxies. Panel (c) compares the empirical CDF with both fitted CDFs over the
full displayed support. Panel (d) focuses on the option exercise region
$x\geq K$ and reports the absolute CDF discrepancy that enters the integrated
pricing-error bound in Proposition~\ref{prop:integrated-cdf-error}.}
\label{fig:crack-distributional-diagnostics}
\end{figure}

The two density approximations reproduce the main right-skewed shape of the
simulated crack-spread distribution. The three-moment proxy captures the
location and overall dispersion, but its fitted density is flatter around the
modal region. The four-moment density follows the height and shoulder of the
Monte Carlo histogram more closely. This difference is economically relevant
because the baseline strike lies to the right of the mode, where the call
value is determined by the remaining upper-tail probability and not by the
center of the distribution alone.

Panel (c) also shows why a global CDF graph, by itself, is not sufficient to
rank the two methods. The three curves are close over most of the support, so
their differences are difficult to distinguish at the scale of a full CDF\@.
Panel (d) magnifies the exercise-relevant region. The four-moment CDF is not
uniformly closer at every point: directly at the baseline strike, the
three-moment CDF happens to be locally closer to the empirical CDF\@. Moving
farther into the exercise region, however, the three-moment discrepancy rises
and remains larger over most of the economically relevant upper tail. The
four-moment discrepancy falls rapidly after the strike and stays comparatively
small through the main part of the tail.

This distinction is important. A call price does not depend only on the CDF
error at the strike. By Proposition~\ref{prop:integrated-cdf-error}, it depends
on the signed area of the CDF difference over the full exercise region
$x\geq K$. The distributional plots therefore explain the baseline result
in Table~\ref{tab:crack-baseline}: the fourth-moment condition does not produce
pointwise dominance, but it gives a better cumulative fit over most of the
exercise region and reduces the integrated tail error. This is consistent with
the decrease in absolute pricing error from $0.3752$ to $0.0583$.

\subsubsection{Accuracy across moneyness and maturity}

A single strike does not show how the approximation behaves across the
exercise region. For each maturity, we therefore use standardized moneyness
\begin{equation}
\label{eq:standardized-moneyness}
z_K
=
\frac{K-E[B_T^{\mathrm{cr}}]}
{\sqrt{\Var(B_T^{\mathrm{cr}})}},
\qquad
K(z)=E[B_T^{\mathrm{cr}}]+z\sqrt{\Var(B_T^{\mathrm{cr}})}.
\end{equation}
The grid is $z\in\{-2,-1.5,\ldots,2\}$. Negative strikes are omitted. The
maturities are 30, 90, and 180 trading days and one year. Each maturity uses
$5\times10^5$ common Monte Carlo paths across strikes.

Table~\ref{tab:crack-accuracy-summary} gives the main error measures. The
short-maturity results are deliberately interpreted with care. At 30 days,
the two methods have almost the same mean absolute error, and the
three-moment method has a slightly smaller RMSE\@. The value of the fourth
moment becomes clearer as maturity increases. At one year, the four-moment
method reduces MAE by 37.5\% and RMSE by 33.2\%. Across all 34 reported
strike--maturity cases, MAE falls from $0.0973$ to $0.0720$, while RMSE falls
from $0.1488$ to $0.1133$.

\begin{table}[htbp]
\centering
\caption{Pricing accuracy across moneyness and maturity}
\label{tab:crack-accuracy-summary}
\resizebox{\textwidth}{!}{%
\renewcommand{\arraystretch}{1.14}
\begin{tabular}{cccccccc}
\toprule
Maturity & Skewness & Kurtosis &
MAE (3M) & MAE (4M) & RMSE (3M) & RMSE (4M) &
Share with 4M better\\
\midrule
30 days  & 0.3548 & 3.3073 & 0.0094 & 0.0093 & 0.0113 & 0.0121 & 44.4\%\\
\rowrule
90 days  & 0.6233 & 3.9838 & 0.0488 & 0.0465 & 0.0587 & 0.0606 & 55.6\%\\
\rowrule
180 days & 0.9003 & 5.1793 & 0.1273 & 0.1052 & 0.1530 & 0.1434 & 62.5\%\\
\rowrule
1 year   & 1.0833 & 6.3262 & 0.2209 & 0.1381 & 0.2581 & 0.1723 & 75.0\%\\
\bottomrule
\end{tabular}}
\end{table}

The maturity pattern is consistent with the distributional role of the fourth
moment. Skewness increases from 0.3548 at 30 days to 1.0833 at one year, and
kurtosis rises from 3.3073 to 6.3262. The fitted four-moment probability also
moves gradually away from the boundary, from $p=0.999992$ at 30 days to
$p=0.996116$ at one year. At short maturities, a one-sided shifted lognormal
already represents the basket reasonably well. At longer maturities, the
basket becomes more asymmetric and heavy-tailed, and the additional moment
condition becomes more useful.

Figure~\ref{fig:crack-moneyness-errors} confirms that there is no pointwise
dominance. The three-moment method is more accurate at some strikes, especially
at short maturities. The four-moment method becomes more consistently accurate
near the center and in the upper-strike region as maturity grows. This is a
more informative conclusion than a claim that the fourth moment always
improves the price.

\begin{figure}[htbp]
\centering
\subfloat[30 trading days]{%
\includegraphics[width=0.48\textwidth]{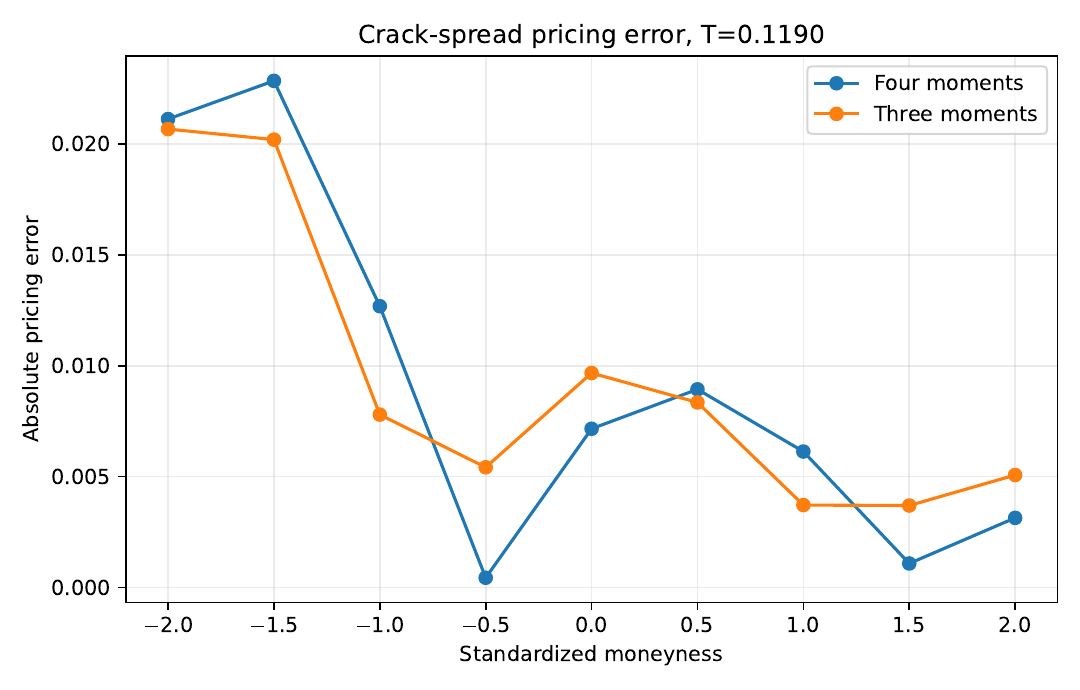}}
\hfill
\subfloat[90 trading days]{%
\includegraphics[width=0.48\textwidth]{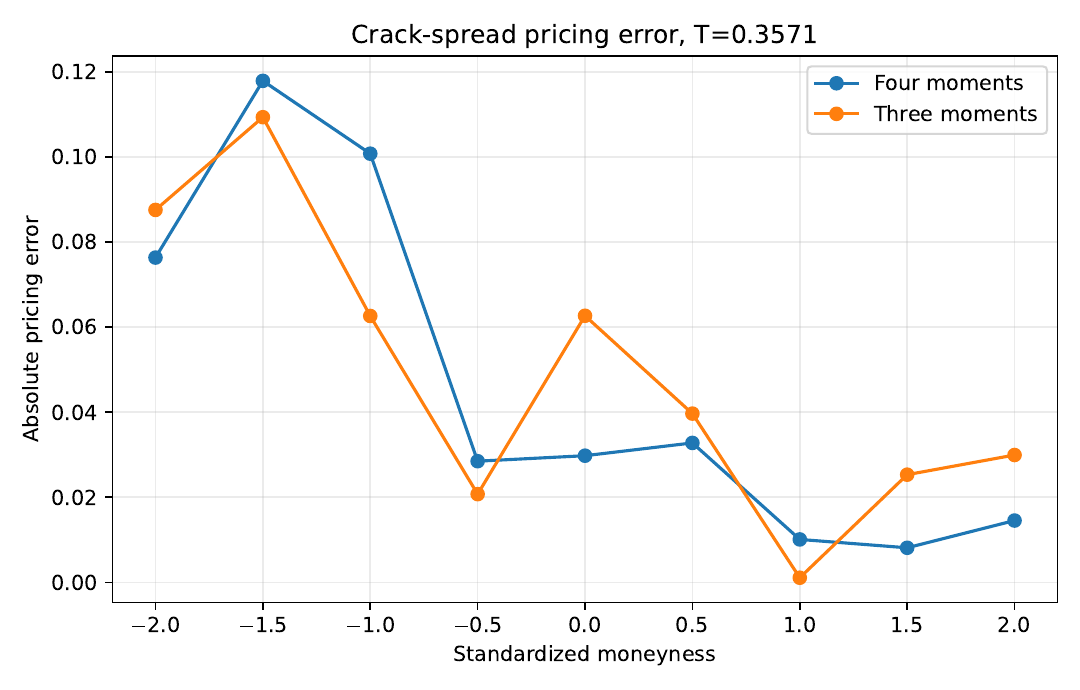}}

\medskip
\subfloat[180 trading days]{%
\includegraphics[width=0.48\textwidth]{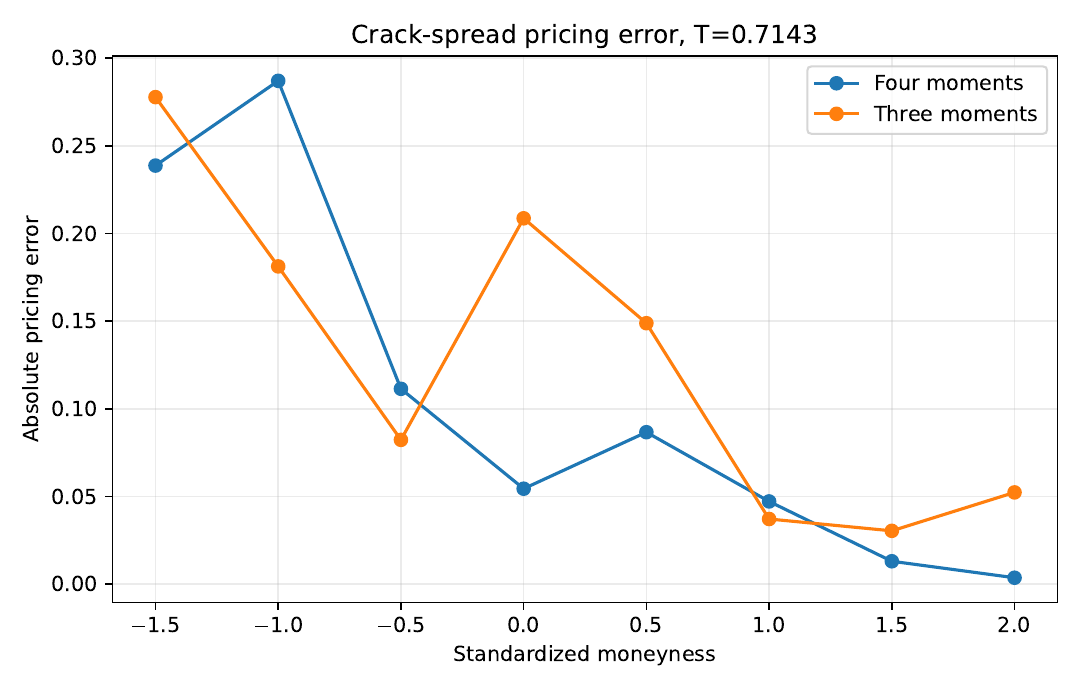}}
\hfill
\subfloat[One year]{%
\includegraphics[width=0.48\textwidth]{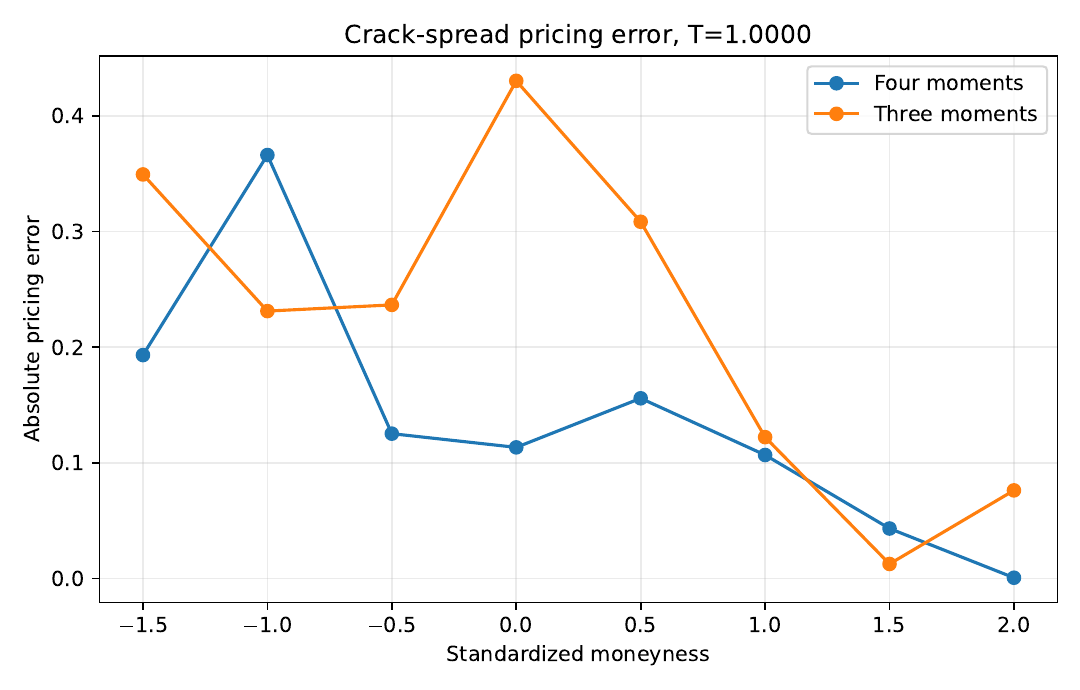}}
\caption{Absolute pricing errors relative to Monte Carlo across standardized
moneyness. The four-moment method does not dominate at every strike, but its
advantage becomes clearer as maturity, skewness, and kurtosis increase.}
\label{fig:crack-moneyness-errors}
\end{figure}

\subsubsection{Tail-error diagnostics and financial consistency}

Proposition~\ref{prop:integrated-cdf-error} implies that the absolute pricing
error is bounded by the discounted integrated CDF discrepancy above the
strike. Figure~\ref{fig:crack-tail-bound} plots the observed pricing errors
against this bound. Every point lies below the 45-degree line, as required by
\eqref{eq:tail-error-bound}. The bound is conservative, but it is informative:
the correlation between the absolute pricing error and the bound is 0.883 for
the four-moment method and 0.824 for the three-moment method.

The mean discounted tail bound is smaller for the four-moment approximation at
every maturity. Relative to the three-moment method, the reduction is 8.4\% at
30 days, 23.2\% at 90 days, 35.6\% at 180 days, and 45.5\% at one year. The
largest numerical discrepancy in the signed identity
\eqref{eq:integrated-cdf-error} is below $0.003$. These results support the
interpretation that the long-maturity price improvement comes from a better
fit in the exercise-relevant tail, not only from a smaller global moment
residual.

\begin{figure}[htbp]
\centering
\includegraphics[width=0.72\textwidth]{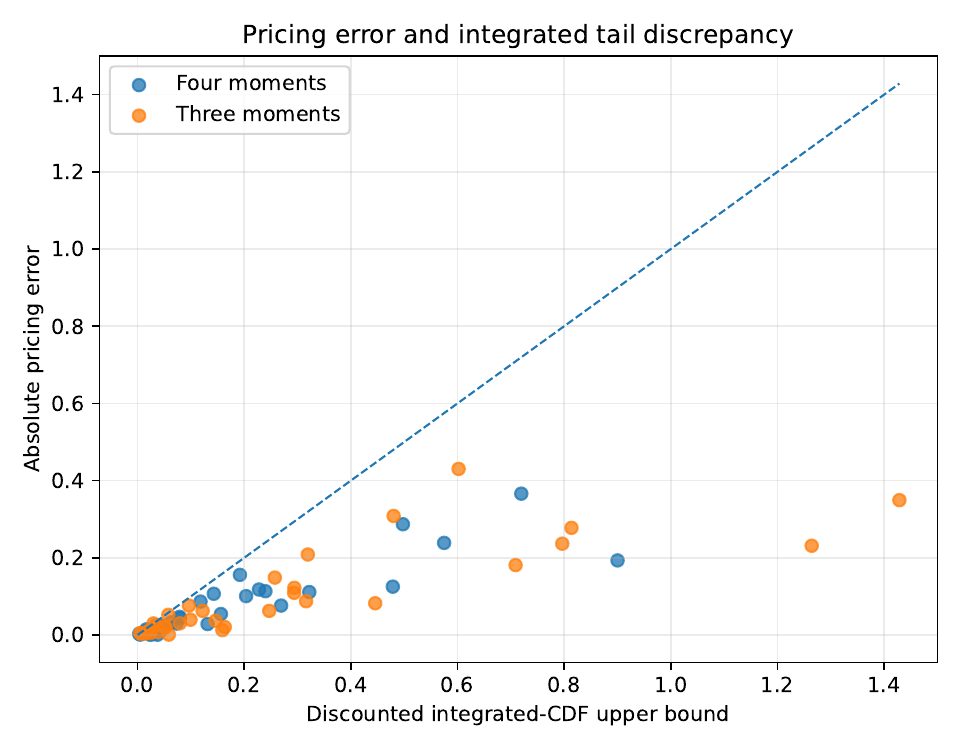}
\caption{Absolute pricing error and the discounted integrated-CDF upper bound.
The dashed line is the 45-degree line. All observations satisfy the theoretical
bound.}
\label{fig:crack-tail-bound}
\end{figure}

The strike grids also provide direct checks of the financial restrictions in
Proposition~\ref{prop:financial-properties}. For both analytical methods and
all four maturities, the numerical study finds no violations of decreasing
call prices, discrete strike convexity, or the forward intrinsic-value lower
bound. The result is consistent with pricing through an admissible proxy
distribution.

\subsubsection{Sensitivity to the historical estimation window}

Historical volatility and correlation estimates can depend on the chosen
sample. We repeat the one-year at-the-money experiment using trailing 1-, 3-,
5-, and 10-year windows. Each case uses $3\times10^5$ Monte Carlo paths.
Table~\ref{tab:crack-window-sensitivity} shows that the four-moment method is
closer to Monte Carlo under every window, although the size of the error changes
with the fitted parameters.

\begin{table}[htbp]
\centering
\caption{Sensitivity to the historical estimation window}
\label{tab:crack-window-sensitivity}
\resizebox{\textwidth}{!}{
\renewcommand{\arraystretch}{1.14}
\begin{tabular}{ccccc}
\toprule
Window & Valid returns & Monte Carlo (s.e.) &
Absolute error (3M) & Absolute error (4M)\\
\midrule
1 year  &  251 & 14.7830 (0.0521) & 0.3874 & 0.0265\\
\rowrule
3 years &  755 & 11.5865 (0.0382) & 0.1345 & 0.0525\\
\rowrule
5 years & 1256 & 13.3589 (0.0456) & 0.1309 & 0.0330\\
\rowrule
10 years& 2511 & 14.7059 (0.0495) & 0.4332 & 0.1093\\
\bottomrule
\end{tabular}
}
\end{table}

The window exercise should not be read as evidence that one historical window
is uniquely correct. It shows instead that the analytical comparison is
stable to materially different parameter estimates: the four-moment method
remains more accurate in each case. At the same time, the variation in the
Monte Carlo benchmark confirms that historical calibration uncertainty is
economically important and is separate from approximation error.

\section{Conclusion}
\label{sec:conclusion}

This work develops a probability-based four-moment framework for basket and
spread option pricing under correlated lognormal dynamics. The construction
first rewrites the exact option value as a linear combination of probabilities.
For a standard basket, these probabilities become CDF values of positive
correlated lognormal sums and are approximated by a shifted lognormal variance
mixture. For a general mixed-sign basket, a signed shifted lognormal proxy
matches the first four moments and yields an analytical call-price formula.
The theoretical analysis states the conditions under which the fitted
parameters define valid distributions and gives a practical procedure for
selecting among numerical roots. For the direct general-basket proxy, the
resulting call price is nonnegative, decreasing, and convex in the strike,
satisfies a Lipschitz bound and put--call parity, and respects the forward
lower bound. The pricing-error identities also show why moment matching alone
is not sufficient. A general-basket call depends on the integrated CDF
discrepancy above the strike, whereas a standard-basket price depends on the
weighted CDF errors of several positive sums at a common threshold.

The standard-basket experiments show that the probability reformulation is an
important part of the method. Applying the three-moment approximation after the
reformulation lowers the mean absolute error from $0.0437$ to $0.0059$, and
matching the fourth moment lowers it further to $0.0012$. For Basket~1, the
four distributional comparisons show that each positive-sum proxy tracks its
benchmark closely around the common threshold $x=1$. These plots support the
probability-level construction, while the exact CDF-error decomposition
explains how the four local approximation errors enter the final option price.
The empirical crack-spread application provides a broader test for a genuine
mixed-sign basket. At the one-year baseline, the four-moment approximation
reduces the absolute pricing error by approximately $84.5\%$ relative to the
three-moment method. Across the complete moneyness and maturity grid, it also
has lower overall MAE and RMSE\@. The improvement is not uniform across strikes.
At short maturities, the two methods are close and the three-moment
approximation is better at some points. The advantage of the fourth moment
becomes clearer as maturity, skewness, and kurtosis increase.
The one-year distributional diagnostics explain this pattern. Although the
three-moment CDF is locally close to the empirical CDF at the baseline strike,
its discrepancy increases farther into the exercise region. The four-moment
CDF remains closer over most of the upper tail, which reduces the integrated
pricing error. The maturity-wide diagnostics support the same interpretation:
the four-moment approximation has a smaller mean tail bound at every maturity,
and every observed price error satisfies the theoretical upper bound. No
violations of strike monotonicity, discrete convexity, or the forward lower
bound are observed.

The empirical exercise uses historical volatilities and correlations estimated
from continuous futures series. It therefore evaluates the analytical
approximation under a fitted lognormal model rather than its ability to match
market option prices. A natural next step is to use contract-matched futures
and option data, estimate risk-neutral inputs from observed option surfaces,
and study hedging errors as well as prices. A second extension is to evaluate
the tail-sensitive mixture parameter out of sample while preserving the
analytical structure of the pricing method.

\FloatBarrier
\bibliographystyle{siam}
\bibliography{reference}

\end{document}